\documentclass[11pt,oneside,reqno]{amsart}

\usepackage{ifthen}
\newboolean{showfigures}
\setboolean{showfigures}{true} 

\usepackage[left]{lineno}
\newboolean{linenumbers}
\setboolean{linenumbers}{false}

\usepackage{adjustbox, algorithm, algpseudocode, amsfonts, amssymb, amsthm, amsmath, array, authblk, bm, booktabs, braket, caption, changepage, comment, dsfont, enumitem, etoolbox, fullpage, graphicx, mathrsfs, mathtools, microtype, multicol, multirow, pifont, pythonhighlight, relsize, subfloat, textcomp, url, xcolor, colortbl, xpatch}
\usepackage[round]{natbib}

\usepackage{algpseudocode, dsfont}
\algrenewcommand\algorithmicrequire{\textbf{Input:}}
\algrenewcommand\algorithmicensure{\textbf{Output:}}

\usepackage{tikz, tikz-cd}
\usetikzlibrary{arrows.meta, positioning, shapes.geometric, decorations.pathreplacing}

\usepackage{subcaption}

\newcommand{\yes}{$\checkmark$}

\newcommand{\na}{\textemdash}

\definecolor{lightgreen}{HTML}{d4edda}
\definecolor{lightgray}{gray}{0.92}

\usepackage[superscript,biblabel,nomove]{cite}

\setlist[enumerate]{itemsep=1ex, topsep=-1ex}
\setlist[itemize]{itemsep=0pt, topsep=0pt}

\usepackage[foot]{amsaddr}

\usepackage{titlesec}

\titleformat{\section}
	{\titlerule\vspace{-2ex}}
	{\thesection.}{1ex}{\centering\scshape}
	[\vspace{.5ex}\titlerule]
	
\titleformat{\subsection}[runin]
  {\normalfont\normalsize\bfseries}{\thesubsection.}{1ex}{}[.]	
\titlespacing*{\subsection}{0pt}{0.0\baselineskip}{0.5ex}

\titleformat{\subsubsection}[runin]
  {\normalfont\normalsize\itshape}{\thesubsubsection.}{1ex}{}	
\titlespacing*{\subsection}{0pt}{0.0\baselineskip}{0.5ex}

\newtheorem{theorem}{Theorem} 

\newtheoremstyle{style}
  {\baselineskip} 
  {0em} 
  {\itshape} 
  {} 
  {\bfseries} 
  {.} 
  {.5em} 
  {} 
\theoremstyle{style}

\newtheorem*{theorem*}{Theorem}
\numberwithin{equation}{section}
\newtheorem{assumption}{Assumption}
\newtheorem{corollary}{Corollary}

\newtheorem*{problem*}{Problem Statement}
\newtheorem{definition}{Definition} 
\newtheorem*{definition*}{Definition}

\newtheorem{lemma}{Lemma}[section]
\newtheorem*{lemma*}{Lemma}

\newtheorem*{prop*}{Proposition}
\newtheorem{remark}{Remark} 

\usepackage{hyperref}
\hypersetup{
	colorlinks=true,
	linkcolor=blue,
	citecolor=black,
	urlcolor=blue}
	
\usepackage[noabbrev,capitalise]{cleveref}
\creflabelformat{equation}{#2\textup{#1}#3}

\crefname{assumption}{Assumption}{Assumptions}
\crefname{condition}{Condition}{Conditions}
\crefname{remark}{Remark}{Remarks}
\crefname{lemma}{Lemma}{Lemmas}
\crefname{corollary}{Corollary}{Corollaries}
\crefname{definition}{Definition}{Definitions}
\numberwithin{equation}{section}

\algrenewcommand\algorithmicrequire{\textbf{Input:}}
\algrenewcommand\algorithmicensure{\textbf{Output:}}

\DeclareMathOperator{\1}{\mathds{1}}

\newcommand{\bbeta}{\bm{\beta}}
\newcommand{\bgamma}{\bm{\gamma}}

\newcommand{\bx}{\bm{x}}
\newcommand{\bX}{\bm{X}}
\newcommand{\css}{\texttt{CausalStabSel}}
\newcommand{\efp}{\mathtt{efp}}
\newcommand{\EFP}{\mathbb{E}(\mathrm{FP})}
\newcommand{\ci}{\perp\!\!\!\!\perp}

\newcommand{\var}{\text{Var}}
\newcommand{\cov}{\text{Cov}}

\title{{\Large C}ausal stability selection}

\author{{\large F}alco {\large J}. {\large B}argagli-{\large S}toffi\textsuperscript{1,$\star$}}
\author{{\large O}mar {\large M}elikechi\textsuperscript{2,$\star$}}

\address{\textsuperscript{1}Department of Biostatistics, University of California, Los Angeles. Email: \url{falco@ucla.edu}}
\address{\textsuperscript{2}Department of Statistical Science, Duke University. Email: \url{omar.melikechi@duke.edu}}

\thanks{\textsuperscript{$\star$}Equal contribution}

\begin{document}

\frenchspacing

\ifthenelse{\boolean{linenumbers}}{\linenumbers}{}

\maketitle

\vspace{-1.25cm}
\begin{abstract}
Identifying covariates that modify treatment effects is a central problem in causal inference. Yet existing data-adaptive procedures do not provide finite-sample control over the expected number of false discoveries, risking spurious findings that fail to replicate. We introduce \textit{causal stability selection}, an algorithm that combines cross-fitted estimation of conditional average treatment effects with integrated path stability selection. The method accommodates arbitrary treatment effect estimators and arbitrary base selectors, and produces a selection set with an explicit, non-asymptotic bound on the expected number of false positives. Under standard causal identifying assumptions and regularity conditions on the base selector, we prove that the estimated selection probabilities converge to their oracle counterparts at the rate of the underlying treatment effect estimator. This establishes a direct connection between treatment effect estimation and effect modifier discovery. We illustrate the method on a randomized trial in oncology and on observational data on maternal smoking and infant birthweight.
\end{abstract}
\vspace{0.25cm}

\section{Introduction}\label{sec:intro}


A large proportion of randomized and observational studies report null average treatment effects. For example, 40\% of trials registered on \texttt{ClinicalTrials.gov} find no significant overall effect \citep{dechartres2016reporting}. Yet null averages need not imply null effects: treatment response may vary systematically with pre-treatment covariates, even when the population mean is zero \citep{gates2019reporting}. The same pattern appears in social policy and public health, where interventions with negligible average impacts often exhibit substantial variation in response \citep{weiss2014conceptual}. Discovering which covariates drive this variation---that is, identifying \textit{effect modifiers}---has both practical and scientific value. Practically, effect modifier discovery is at the foundation of treatment targeting for precision interventions \citep{kosorok2019precision}. Scientifically, identifying effect modifiers helps characterize populations and subgroups in which treatments are most and least effective \citep{athey2016recursive}. Identifying effect modifiers can therefore reveal clinical or policy value in interventions otherwise deemed ineffective.

Formally, effect modifiers are defined in terms of the \textit{conditional average treatment effect} (CATE),
\begin{align}\label{eq:cate}
\tau(\bx) &= \mathbb{E}\left[Y(1) - Y(0) \mid \bX = \bx\right],
\end{align}
where $\bX=(X_1,\dots,X_p)\in\mathbb{R}^p$ are the pre-treatment covariates and $Y(1)$ and $Y(0)$ are potential outcomes under treatment and control, respectively \citep{rubin1974estimating}. When $\tau$ is non-constant in $\bx$, the treatment effect is said to be \textit{heterogeneous}. Many approaches have been developed to estimate $\tau$; however, a method may estimate $\tau$ accurately without revealing which covariates drive heterogeneity. In this work, our objective is effect modifier discovery, that is, to identify the covariates that explain variation in $\tau$ (this task is formally defined in Section~\ref{sec:setup}).  

While CATE estimation has received significant attention in recent years, several methods aimed specifically at effect modifier discovery have also been introduced. Many of these focus on constructing variable importance measures to quantify each covariate's contribution to treatment effect heterogeneity. For example, \citet{hines2025variable} develop model-free importance measures based on efficient influence functions, \citet{benard2023variable} derive importance measures from causal forests, and \citet{boileau2025nonparametric} introduce a nonparametric framework for defining effect modifier importance and establish asymptotic properties of their proposed estimators. While these approaches provide valuable tools for quantifying the contribution of each covariate to heterogeneity---and some include inferential procedures that can inform selection---they are not primarily designed to identify a parsimonious set of genuine effect modifiers with error rate guarantees. In a separate line of work more closely aligned with our objective, \citet{zhao2022selective} propose a two-stage procedure that first estimates nuisance functions related to the CATE, then applies the lasso \citep{tibshirani1996regression} to select a parsimonious model for effect modification. 

Despite these advances, existing methods do not provide finite-sample bounds on the expected number of selected covariates that are spurious effect modifiers. As we illustrate below, the absence of such control can lead to severely inflated false positive rates, undermining the reliability and replicability of downstream analyses and targeted interventions \citep{benjamini1995controlling, yu2013stability}. Achieving false discovery control while maintaining power to detect genuine effect modifiers is the central challenge we address.


\subsection{Illustrative example}\label{sec:illustration}


To illustrate the challenge of effect modifier discovery, we present the results of a simulation study in which data are generated from the linear model
\begin{align}\label{eq:linear}
Y &= \bX^\top\bbeta + Z\bX^\top\bgamma + \varepsilon,
\end{align}
where $Z\in\{0,1\}$ is a binary treatment and $\varepsilon$ is random noise satisfying $\mathbb{E}[\varepsilon\mid\bX,Z]=0$. The coefficients $\bbeta,\bgamma\in\mathbb{R}^p$ capture effects of $\bX$ on $Y$ that are unrelated and related to treatment, respectively, and their nonzero coordinates $\mathcal{P}=\{j : \beta_j\neq 0\}$ and $\mathcal{E}=\{j:\gamma_j\neq 0\}$ correspond to \textit{prognostic variables} and \textit{effect modifiers}. Under standard assumptions formalized in \cref{sec:setup},
\begin{align*}
Y(1) - Y(0) &= \bX^\top\bgamma + \varepsilon(1) - \varepsilon(0),
\end{align*}
where $Y(z)$ and $\varepsilon(z)$ are the outcome and noise terms when $Z=z$. Hence the CATE~\eqref{eq:cate} satisfies
\begin{align*}
\tau(\bX) &= \mathbb{E}\left[Y(1) - Y(0) \mid \bX\right]
    = \bX^\top\bgamma.
\end{align*} 

\cref{fig:illustration} shows the results of this simulation study averaged over 200 trials. In each trial, $n=1000$ samples are drawn independently according to \cref{eq:linear}, where $p=100$ and $\bX\sim\mathcal{N}(0,\bm{\Sigma})$ with $\bm{\Sigma}\in\mathbb{R}^{100\times 100}$ given by $\Sigma_{ij}=0.5^{\lvert i-j\rvert}$. Prior to each trial, $\bbeta$ and $\bgamma$ are assigned 10 nonzero entries selected uniformly at random from $\{1,\dots,100\}$, with values drawn uniformly from $[-1,-0.5]\cup[0.5,1]$. The overall signal-to-noise ratio, SNR = $\var(\bX^\top\bbeta + Z\bX^\top\bgamma)/\var(\varepsilon)$, is set to 1.

If $\tau$ were known, identifying effect modifiers would amount to standard supervised variable selection with predictors $\bX$ and response $\tau(\bX)$. The panels on the right side of \cref{fig:illustration} show that the lasso and integrated path stability selection \citep{melikechi2026integrated} perform especially well in this case, identifying all 10 effect modifiers on average while maintaining very low FDR, while the Knockoffs method \citep{barber2015controlling} also performs well once the nominal FDR exceeds 0.1. This \textit{oracle setting} is untenable in practice, however, because each unit receives only one of the two potential treatments, never both. Consequently, $\tau(\bX)$ must be estimated from data. 

\begin{figure}
\includegraphics[width=\textwidth]{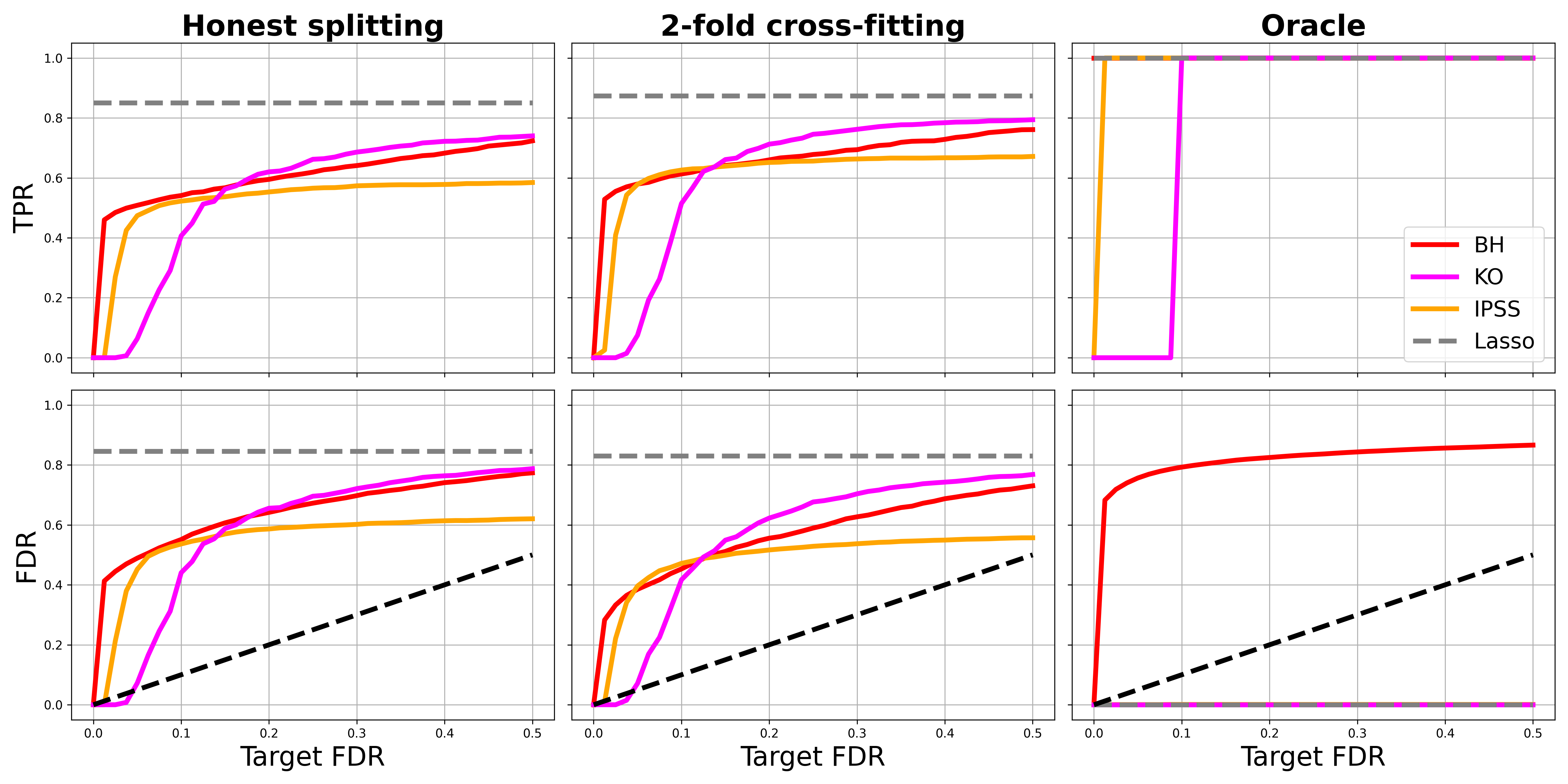}%
\caption{\textit{Performance of variable selection methods for effect modifier discovery.} \textbf{Top row:} True positive rate (TPR), defined as the proportion of identified effect modifiers among all true effect modifiers, averaged over $200$ simulation trials. \textbf{Bottom row:} False discovery rate (FDR), defined as the proportion of false discoveries among all selected covariates, averaged over $200$ trials. Black diagonal dashed lines in the bottom panels represent perfect FDR control; points on or below this line indicate successful control at the nominal level. \textbf{Left panels:} Results obtained under honest splitting. \textbf{Middle panels:} Results obtained under 2-fold cross-fitting. \textbf{Right panels:} Results obtained under ``oracle setting'' in which $\tau$ is known; FDR control fails for BH in this setting largely due to correlation between covariates. \textbf{Acronyms}: IPSS: integrated path stability selection; KO: knockoffs; BH: Benjamini-Hochberg; Lasso: Least Absolute Shrinkage and Selection Operator.}
\label{fig:illustration}
\end{figure}

Two general approaches for handling the unknown function $\tau$ are \textit{honest splitting} and \textit{$K$-fold cross-fitting}, both of which have been used extensively for CATE estimation \citep{athey2016recursive,chernozhukov2018double,kennedy2023towards} and, less commonly, for effect modifier/subgroup discovery \citep{zhao2022selective,bargagli2020causal, bargagli2025heterogeneous,huang2025distilling}. In the context of effect modifier discovery, honest splitting proceeds as follows: samples are partitioned into disjoint subsets $A$ and $A^c$; samples in $A$ are used to construct an estimate $\hat{\tau}_A$ of $\tau$, and variable selection is performed using the pairs $\{(\bx_i,\hat{\tau}_A(\bx_i)) : i\in A^c\}$. This approach reduces overfitting since $\hat{\tau}_A$ is constructed from samples in $A$ while the selection step uses those in $A^c$. However, honest splitting only uses half of the available observations for variable selection, which can reduce power. $K$-fold cross-fitting seeks to ameliorate this issue by partitioning samples into $K$ folds $\{A_1,\dots,A_K\}$. For each $k=1,\dots,K$, an estimate $\hat{\tau}^{(k)}$ is constructed using all samples not in $A_k$, and variable selection is performed on the aggregated pairs $\{(\bx_i,\hat{\tau}^{(k)}(\bx_i)) : i \in A_k, k=1,\dots,K\}$. 

Despite their appeal, the left and central panels shown in \cref{fig:illustration} indicate that neither lasso nor any of the variable selection methods with false discovery control are able to keep the false discovery rate (FDR) below nominal levels when combined with either honest splitting or cross-fitting. Lasso, which does not provide explicit false discovery control (and hence does not vary with the target FDR), makes a high number of false discoveries in both cases: on average, over $80\%$ of the covariates it selects are not effect modifiers. Moreover, results for this same simulation design indicate that increasing the number of folds $K$ does little to resolve this issue (\cref{fig:crossfitting}).


\subsection{Contributions}


We present a formal framework for effect modifier discovery and a method, \textit{Causal Stability Selection} (\texttt{CausalStabSel}), for selecting effect modifiers. The method combines CATE estimation and cross-fitting from modern causal inference \citep{fan2022estimation,kennedy2023towards} with stability selection \citep{meinshausen2010stability,melikechi2026integrated}. To our knowledge, \texttt{CausalStabSel}---which is compatible with any choice of CATE estimator and base selection algorithm---is the first effect modifier discovery method with finite-sample control over the expected number of false discoveries. We also prove that the set of variables selected by \texttt{CausalStabSel} converges to the set obtained if the true CATE were known (the oracle solution) at a rate governed by the CATE estimation error, establishing a direct connection between prediction and selection in the causal setting. Furthermore, by aggregating selections across repeated subsamples, \texttt{CausalStabSel} improves stability while remaining fast and easily parallelized, keeping it computationally tractable even in high dimensions. Finally, our proposed framework requires standard causal assumptions (SUTVA and strong ignorability) and regularity conditions on the base selector, making it applicable to both randomized trials and observational studies.


\subsection{Paper structure} 


The remainder of the paper is organized as follows. 
Section~\ref{sec:setup} formally defines the effect modifier discovery problem. Section~\ref{sec:method} introduces \texttt{CausalStabSel}. Section~\ref{sec:theory} establishes theoretical guarantees. Section~\ref{sec:simulations} presents simulation studies, and Section~\ref{sec:realdata} applies the method to real data. Section~\ref{sec:discussion} concludes. An open-source package implementing our method and the applied examples is provided.


\section{Problem formulation}\label{sec:setup}


In \Cref{subsec:notation}, we introduce notation and the potential outcomes framework of causal inference under which we operate. In \Cref{subsec:hte}, we formally define the central quantities of our investigation---effect modifiers---and the task of effect modifier discovery.


\subsection{Notation and causal framework}\label{subsec:notation}


Uppercase letters denote random variables and lowercase letters represent their realized values. Bold letters represent vectors and matrices. For any positive integer $n$, define $[n]=\{1,\dots,n\}$. We assume that our observed data are realizations $(\bx_i,y_i,z_i)$ of $n$ independent random vectors $(\bX_i,Y_i,Z_i)$ each with the same distribution as $(\bX,Y,Z)$, where $\bX\in\mathbb{R}^p$ is a vector of covariates, $Z\in\{0,1\}$ is a binary treatment, and $Y\in\mathbb{R}$ is the observed outcome. For generic random variables $A$, $B$, and $C$, we write $A \ci B \mid C$ when $A$ and $B$ are conditionally independent given $C$, and $A \not\ci B \mid C$ otherwise. For $\bX\in\mathbb{R}^p$, define $\bX_{-j}=(X_1,\dots,X_{j-1},X_{j+1},\dots,X_p)$. Variables $X_j$ are often referred to by their indices, e.g., ``$j$ is an effect modifier'' is shorthand for ``$X_j$ is an effect modifier.'' Finally, $\1(\cdot)$ denotes the indicator function: $\1(A)=1$ if event $A$ is true, and $\1(A)=0$ otherwise.

\begin{assumption}[Stable Unit Treatment Value Assumption (SUTVA)]
\label{ass:SUTVA}
For all $i\in[n]$, the potential outcome $Y_i$ satisfies:
\begin{itemize}[leftmargin=2em, itemsep=0.25em]
\item[(i)] \textbf{No interference}: $Y_i(\mathbf{z}) = Y_i(z_i)$ for any $\mathbf{z} = (z_1,\ldots,z_n)$;
\item[(ii)] \textbf{Consistency}: If $Z_{i} = z$, then $Y_{i} = Y_{i}(z)$ almost surely.
\end{itemize}
\end{assumption}

Part (i) states that each unit's outcome does not depend on treatments assigned to the other units, and hence that the potential outcomes $Y_i(0)$ and $Y_i(1)$ are well-defined. 
Part (ii) ensures each treatment level has a single version, so the observed outcome satisfies $y_i = Y_i(z_i)$. 
The fundamental problem of causal inference is that only one of $Y_i(0)$ or $Y_i(1)$ is observed for each $i$ \citep{holland1986statistics}.

\begin{assumption}[Strong Ignorability]
\label{ass:ignorability}
The treatment assignment satisfies:
\begin{itemize}[leftmargin=2em, itemsep=0.25em]
\item[(i)] \textbf{Unconfoundedness}: $(Y(0), Y(1)) \ci Z \mid \bX$; 
\item[(ii)] \textbf{Overlap/Positivity}: $0 < P(Z = 1 \mid \bX = \bx) < 1$ for all $\bx$ in the support of $\bX$.
\end{itemize} 
\end{assumption}

Strong ignorability ensures that, conditional on covariates, treatment assignment is independent of the potential outcomes (unconfoundedness) and that each unit has a positive probability of receiving either treatment (overlap) \citep{rosenbaum1983central}.

Under Assumption~\ref{ass:SUTVA}, the observed outcome satisfies
\begin{equation}\label{eq:decomposition}
\begin{split}
Y 
= Y(0) + Z(Y(1) - Y(0))
= \mu(\bX) + Z\tau(\bX) + \varepsilon,
\end{split}
\end{equation}
where $\mu(\bX) = \mathbb{E}[Y(0) \mid \bX]$ is the prognostic function, $\tau(\bX) = \mathbb{E}[Y(1) - Y(0) \mid \bX]$ is the CATE, and $\varepsilon = Y(0) - \mu(\bX) + Z(Y(1) - Y(0) - \tau(\bX))$ satisfies $\mathbb{E}[\varepsilon \mid \bX, Z] = 0$.
Under Assumption~\ref{ass:ignorability}, the prognostic and CATE functions satisfy
\begin{align*}
\mu(\bX) &= \mathbb{E}[Y(0) \mid \bX] 
             = \mathbb{E}[Y \mid Z = 0, \bX], \\
\tau(\bX) &= \mathbb{E}[Y(1) - Y(0) \mid \bX] 
             = \mathbb{E}[Y \mid Z = 1, \bX] 
             - \mathbb{E}[Y \mid Z = 0, \bX].
\end{align*}


\subsection{Treatment effect modifiers}\label{subsec:hte}


The CATE captures how treatment effects vary with $\bX$, while $\mu(\bX)$ captures associations between $\bX$ and $Y$ unrelated to treatment. 
Formally, we define prognostic variables and effect modifiers as follows.

\begin{definition}[\textbf{Prognostic Variables and Effect Modifiers}]\label{def:prognostic_and_effectmodifiers}
A covariate $X_j$ is:
\begin{itemize}[leftmargin=2em, itemsep=0.25em]
\item[(i)] A \textbf{prognostic variable} if $\mu(\bX) \not\ci X_j \mid \bX_{-j}$, i.e., $X_j$ provides information about the baseline outcome not captured by other covariates;
\item[(ii)] A \textbf{treatment effect modifier}, or, more simply, an \textbf{effect modifier}, if $\tau(\bX) \not\ci X_j \mid \bX_{-j}$, i.e., $X_j$ provides information about treatment effect heterogeneity not captured by other covariates.
\end{itemize}
We denote the sets of prognostic variables and effect modifiers by $\mathcal{P} := \{ j : \mu(\bX) \not\ci X_j \mid \bX_{-j} \}$ and $\mathcal{E} := \{ j : \tau(\bX) \not\ci X_j \mid \bX_{-j} \}$, respectively.
\end{definition}


\begin{remark}[Interpretation of Definition~\ref{def:prognostic_and_effectmodifiers}]
Since $\tau(\bX)$ and $\mu(\bX)$ are deterministic functions of $\bX$, the conditional independence statements in Definition~\ref{def:prognostic_and_effectmodifiers} have a direct functional interpretation. The statement $\tau(\bX) \not\ci X_j \mid \bX_{-j}$ is equivalent to $\tau$ not being almost-surely a function of $\bX_{-j}$ alone; that is, there exist values $\bx_{-j}$ at which $\tau$ varies with $x_j$. An analogous interpretation applies to $\mu$ and prognostic variables.
\end{remark}

The conditional independence formulation in Definition~\ref{def:prognostic_and_effectmodifiers} extends prior characterizations developed for specific settings (for example, coefficient sparsity in linear models \citep{tian2014simple}) to a general, model-agnostic form applicable to arbitrary $\mu$ and $\tau$. In the linear setting, \cref{eq:linear}, where $\mu(\bX) = \bX^\top\bbeta$ and $\tau(\bX) = \bX^\top\bgamma$, these reduce to $\mathcal{P} = \{j : \beta_j \neq 0\}$ and $\mathcal{E} = \{j : \gamma_j \neq 0\}$. Furthermore, we note that $\mathcal{P}$ and $\mathcal{E}$ are not mutually exclusive; that is, a covariate may be both prognostic and an effect modifier.

The following lemma shows that our definition of effect modifier in \cref{def:prognostic_and_effectmodifiers} is equivalent to variance-based variable importance measures of heterogeneity \citep{hines2025variable,williamson2023general,shin2025treatment}. 
However, the latter require that $\tau(\bX)$ has a finite second moment, making our characterization in terms of conditional independence slightly more general.

\begin{lemma}[Equivalent Characterizations of Effect Modifiers]
\label{lemma:equivalent}
Assume $\mathbb{E}(\tau(\bX)^2) < \infty$ and define $\tau_j(\bX)=\mathbb{E}[Y(1) - Y(0) \mid \bX_{-j}]$. Then $X_j$ is an effect modifier if and only if $\var(\tau_j(\bX)) \neq \var(\tau(\bX))$. 
\end{lemma}

\begin{proof}
If $X_j$ is not an effect modifier, that is, if $\tau(\bX) \ci X_j \mid \bX_{-j}$, then
\begin{align*}
\tau(\bX) &= \mathbb{E}[\tau(\bX) \mid \bX_{-j}] 
    = \mathbb{E}[Y(1) - Y(0) \mid \bX_{-j}] 
    = \tau_j(\bX)
\end{align*}
almost surely, where the second equality holds by the tower property of conditional expectation. 
Hence $\var(\tau(\bX))=\var(\tau_j(\bX))$. 
Conversely, if $\var(\tau(\bX))=\var(\tau_j(\bX))$, then
\begin{align*}
\var(\tau_j(\bX)) &= \var(\tau(\bX)) \\
    &= \var(\mathbb{E}[\tau(\bX)\mid \bX_{-j}]) + \mathbb{E}(\var[\tau(\bX) \mid \bX_{-j}]) \\
    &= \var(\tau_j(\bX)) + \mathbb{E}(\var[\tau(\bX) \mid \bX_{-j}]).
\end{align*}
Hence $\mathbb{E}(\var[\tau(\bX) \mid \bX_{-j}]) = 0$, implying that $\var[\tau(\bX) \mid \bX_{-j}]=0$ almost surely. 
Thus, $\tau(\bX)$ is measurable with respect to the $\sigma$-algebra generated by $\bX_{-j}$, so $\tau(\bX) \ci X_j \mid \bX_{-j}$.
\end{proof}

With our definition of effect modifier in hand, we now state the principal goal of this work.
\begin{problem*}[\textbf{Effect Modifier Discovery}]
Given $n$ samples $(\bx_i,y_i,z_i)\in\mathbb{R}^p\times\mathbb{R}\times\{0,1\}$, construct an estimator $\hat{\mathcal{E}}$ of $\mathcal{E}$ that controls the expected number of false positives, $\EFP=\mathbb{E}|\hat{\mathcal{E}}\cap\mathcal{E}^c|$.
\end{problem*}

\cref{fig:illustration} showed that many existing methods fail at this task. In the next section, we introduce \texttt{CausalStabSel}, a method designed for effect modifier discovery with false discovery control.

\section{Causal Stability Selection}\label{sec:method}


Rather than apply a generic variable selection algorithm once to the full dataset, \css{}---and stability selection more broadly---aggregates selections across multiple random subsamples of the data, yielding more stable results and finite-sample false discovery control \citep{meinshausen2010stability}. However, existing versions of stability selection assume the outcome of interest is directly observed, making them ill-equipped for effect modifier discovery. To overcome this limitation, \css{} incorporates CATE estimation into the subsampling scheme, repeatedly using one subset of samples to estimate the CATE and the remaining samples for selection.

The rest of this section is dedicated to describing the \css{} algorithm. In \cref{sec:selection_probabilities}, we introduce the central quantities underlying \css{}, \textit{selection probabilities}, which quantify the probability that a covariate is an effect modifier. In \cref{sec:selection_criteria}, we explain how selection probabilities are used to construct the final \css{} estimate of the true set of effect modifiers, $\mathcal{E}$, and in \cref{sec:parameters}, we discuss the inputs to \css{}.

\begin{algorithm}
\caption{Cross-fitted Estimation of Selection Probabilities}\label{alg:selection_probabilities}
\begin{algorithmic}[1]
\Require Data $\{(\bx_i,y_i,z_i)\}_{i=1}^n$, base selector $\hat{S}_\lambda$, CATE estimator $\hat{\tau}$, parameter grid $\Lambda$, number of iterations $B$, number of subsamples $m\leq \lfloor n/2\rfloor$.
\For{$b=1,\ldots,B$}
    \State Randomly select disjoint subsets $A_{2b-1}, A_{2b}\subseteq[n]$ of size $m$.
    \State Train $\hat{\tau}$ on samples $\{(\bx_i,y_i,z_i) : i \in A_{2b-1}^c\}$. Denote the model by $\hat{\tau}_{A_{2b-1}^c}$.
    \State Predict $\hat{\tau}_{A_{2b-1}^c}(\bx_i)$ for all $i \in A_{2b-1}$.
    \For{$\lambda \in \Lambda$}
        \State Evaluate $\hat{S}_\lambda(\{(\bx_i, \hat{\tau}_{A_{2b-1}^c}(\bx_i)) : i \in A_{2b-1}\})$.
    \EndFor
    \State Repeat steps 3--7 with $A_{2b}$ in place of $A_{2b-1}$.
\EndFor
\end{algorithmic}
\begin{flushleft}
\textbf{Output:} Estimated selection probabilities 
$\hat{\pi}_{j,m}(\lambda) = \frac{1}{2B} \sum_{b=1}^{2B} \1\left(j \in \hat{S}_\lambda(\{(\bx_i, \hat{\tau}_{A_b^c}(\bx_i)) : i \in A_b\})\right)$.
\end{flushleft}
\end{algorithm}

\subsection{Selection probabilities}\label{sec:selection_probabilities}

Let $\hat{S}_\lambda$ be any supervised variable selection algorithm parametrized by $\lambda\in\mathbb{R}$; that is, given observed covariate--response pairs $\{(\bx_i,\tau_i) : i\in A\}$ for some index set $A$, $\hat{S}_\lambda(\{(\bx_i,\tau_i) : i\in A\})\subseteq[p]$ is an estimate of the covariates that directly relate to the response. In our setting, the responses $\tau_i = \tau(\bx_i)$ are the unobserved CATE values, and the relevant covariates are the effect modifiers, $\mathcal{E}$. Specific examples of $\hat{S}_\lambda$ are provided in \cref{sec:parameters}. 

For a subset $A\subseteq[n]$ of size $m$, the \textit{selection probability} $\pi_{j,m}(\lambda)$ of variable $j$ is
\begin{align}\label{eq:selection_probabilities}
\pi_{j,m}(\lambda) &= \mathbb{P}\left(j \in \hat{S}_\lambda(\bX_A,\tau(\bX_A))\right).
\end{align}
We regard these selection probabilities as oracle quantities because they correspond to the idealized setting in which $\tau$ is known. Note that under the independent samples assumption, $\pi_{j,m}$ depends on $A$ only through the subsample size $m$.

\css{} consists of two key steps: (i) estimating the selection probabilities $\pi_{j,m}(\lambda)$ from observed data, and (ii) using the estimated selection probabilities to construct an estimator of $\mathcal{E}$. \cref{alg:selection_probabilities} describes step (i): in each of $B$ subsampling steps, a random subset $A_{2b-1}\subseteq[n]$ of size $m$ is selected, and a user-specified CATE estimation algorithm $\hat{\tau}$ is trained using the held-out samples $A_{2b-1}^c = [n]\setminus A_{2b-1}$. The resulting model, $\hat{\tau}_{A_{2b-1}^c}$, is then used to predict CATE values for samples in $A_{2b-1}$ and $\hat{S}_\lambda$ is applied to the pairs $\{(\bx_i,\hat{\tau}_{A_{2b-1}^c}(\bx_i)) : i\in A_{2b-1}\}$. The final estimated selection probability, $\hat{\pi}_{j,m}(\lambda)$ for feature $j$ at $\lambda$, is the proportion of times $j$ is selected across all $2B$ subsamples. This complementary pairs design, which consists of drawing two disjoint subsets per iteration rather than one, has been shown to moderately improve selection performance relative to the original stability selection algorithm \citep{shah2013variable}.

Using disjoint samples for CATE estimation and selection in \cref{alg:selection_probabilities} is critical: when the same samples are used for both, the prediction $\hat{\tau}(\bx_i)$ depends on $\bx_i$ through the training process, inducing an artifactual dependence between $\bx_i$ and $\hat{\tau}(\bx_i)$ that the base selector mistakes for genuine effect modification. Cross-fitting breaks this dependence by ensuring $\hat{\tau}_{A^c}$ is independent of $\bx_i$ for $i \in A$. Furthermore, as shown by the violations of FDR control in \cref{fig:illustration}, splitting the data only once prior to applying stability selection or other variable selection methods does little to mitigate this problem. Only by \textit{repeated} cross-fitting---that is, training $\hat{\tau}$ on samples in $A^c$ and selecting on samples in $A$ across multiple iterations---is this source of bias controlled. We formalize this in \cref{sec:theory}.

\subsection{Constructing the selection set}\label{sec:selection_criteria}

After estimating the selection probabilities via \cref{alg:selection_probabilities}, the second and final step of \css{} is to construct an estimate of $\mathcal{E}$ so as to control the expected number of false positives, $\mathbb{E}(\mathrm{FP})$. In the original stability selection paper, \citet{meinshausen2010stability} achieve this by considering the estimate
\begin{align}\label{eq:ss}
\hat{\mathcal{E}}_{\text{SS}} &= \left\{j : \max_{\lambda\in\Lambda}\hat{\pi}_{j,m}(\lambda) \geq \gamma\right\},
\end{align}
where $\Lambda\subseteq\mathbb{R}$ is a set of $\lambda$ values and $\gamma$ is a user-specified threshold. However, this selection criterion is known to produce overly-conservative estimates, due in part to its relatively loose theoretical upper bound on $\mathbb{E}(\mathrm{FP})$ \citep{alexander2011stability,melikechi2026integrated}. In \citet{melikechi2026integrated}, a significantly tighter upper bound on $\mathbb{E}(\mathrm{FP})$ is obtained via the selection criterion
\begin{align}\label{eq:ipss}
\hat{\mathcal{E}}_{\gamma} &= \left\{j : \int_\Lambda f\left(\hat{\pi}_{j,m}(\lambda)\right)\, \mu(d\lambda) \geq \gamma\right\},
\end{align}
where $f(x) = (2x - 1)^3\1(x \geq 1/2)$ and $\mu$ is an arbitrary probability measure on $\Lambda$ whose specification is discussed in \cref{sec:parameters}. This approach, called \textit{integrated path stability selection} (IPSS), achieves this tighter upper bound on $\mathbb{E}(\mathrm{FP})$ by transforming (via $f$) and integrating (rather than maximizing) the selection probabilities over $\Lambda$, leading to more true positives than other versions of stability selection while maintaining false positive control. For this reason, we use the IPSS criterion \eqref{eq:ipss} when implementing \css{}. In \cref{sec:theory_fdr}, we prove that an upper bound on $\mathbb{E}(\mathrm{FP})$ similar to the original IPSS bound holds for \css{}.

Another benefit of IPSS is that the threshold $\gamma$ does not need to be specified by the user. Instead, the theoretical upper bound on $\mathbb{E}(\mathrm{FP})$ is used to assign an \textit{expected false positive} (efp) score, $\efp(j)$, to each variable $j$ \citep{melikechi2025nonparametric}. Given a user-specified target $\mathbb{E}(\mathrm{FP})$ level $t$, the final selection set $\hat{\mathcal{E}} = \{j : \efp(j) \leq t\}$ has at most $t$ false positives on average. For FDR control, we select the largest set satisfying $t / |\hat{\mathcal{E}}| \leq \alpha$ for a nominal FDR level $\alpha$.

\subsection{Implementation}\label{sec:parameters}

\css{} involves several inputs: the variable selection algorithm, or \textit{base selector}, $\hat{S}_\lambda$, the CATE estimator $\hat{\tau}$, the subsample size $m$, the parameter set $\Lambda$ and probability measure $\mu$ on $\Lambda$, and the number of iterations $B$. The CATE estimator is unique to \css{} relative to other stability selection methods, and $m$ plays a different role here than in standard stability selection: smaller subsamples leave more data for CATE estimation, and we show in Section 4 that this trade-off is asymptotically favorable. Both are further discussed below and analyzed in detail in \cref{sec:theory}. The remaining inputs serve identical purposes in \css{} and in IPSS. We briefly describe these below, referring the reader to \citet{melikechi2026integrated} and \citet{melikechi2025nonparametric} for details.

\css{} is compatible with any base selector $\hat{S}_\lambda$. Two broad classes are (i) penalized regression and (ii) thresholding variable importance scores. The former consists of regularized parametric models such as the lasso \citep{tibshirani1996regression} and adaptive lasso \citep{zou2006adaptive}, in which case $\lambda\geq 0$ is the regularization strength. The latter consists of selectors $\hat{S}_\lambda = \{j : \Phi_j > \lambda\}$, where $\Phi_j$ measures the importance of variable $j$ computed via any algorithm---such as tree-based or deep learning models---and $\lambda$ is the importance threshold. In this work, we use lasso-based selection and mean decrease impurity (MDI) from XGBoost \citep{chen2016xgboost} as our base estimators. When combined with IPSS, MDI from XGBoost exhibits strong selection performance and computational efficiency relative to other importance scores across a wide range of settings, most notably when there are nonlinear relationships in the data \citep{melikechi2025nonparametric}. We note that the Lipschitz condition on importance scores introduced in Section~\ref{sec:theory} (Assumption~\ref{ass:lipschitz}) is most readily verified for linear scores such as lasso coefficients; for tree-based importance measures including MDI, the condition is taken as a high-level regularity assumption.

Similar to the base selector, \css{} is compatible with any CATE estimation method. In this work, we consider three meta-learning approaches: T- and X-learners \citep{kunzel2019metalearners}, and the DR-learner \citep{kennedy2023towards}. Each combines one or more supervised regression models---its \textit{base learners}---to form a CATE estimate. On the theoretical side, the performance of \css{} is governed by the CATE estimation error $\mathbb{E}\|\hat{\tau}_{A^c} - \tau\|$; when this error vanishes at the parametric rate $n^{-1/2}$, the \css{} selection probabilities are asymptotically unbiased (\cref{cor:cate_rate}), the \css{} selection set converges to its oracle counterpart (\cref{cor:oracle_convergence}), and the false discovery bound matches that of the oracle setting (\cref{thrm:oracle_fdr}). 
On the empirical side, we compare all three meta-learners, each implemented with both ridge regression and gradient boosting as base learners (with lasso as the \css{} base selector in the linear setting and XGBoost importance scores in the nonlinear setting). \cref{fig:cate_comparison} reports the results. The three ridge-based implementations produce nearly identical \css{} performance: strong in the linear setting, where the model is correctly specified, and poor in the nonlinear setting. The XGBoost-based implementations differ modestly, with the DR-learner achieving better FDR control than the T- and X-learners. All three perform well in the nonlinear setting and moderately worse than the ridge-based estimators in the linear setting, and their computational requirements are similar in practice. Combined with its theoretical properties, this leads us to recommend the DR-learner as the default CATE estimator in \css{}.

For IPSS and other forms of stability selection, the subsample size $m$ is typically fixed at $\lfloor n/2 \rfloor$. For \css{}, however, there is benefit to choosing smaller values when $n$ is large, since the $n - m$ held-out samples are used for CATE estimation which, in our experience, tends to benefit more than selection from additional samples. In \cref{sec:theory}, we formalize this claim, proving that \css{} achieves oracle performance asymptotically as $n, B \to \infty$ provided that $m = o(\sqrt{n})$. For smaller sample sizes where the cost of reducing $m$ below $\lfloor n/2 \rfloor$ outweighs the benefit, we recommend $m = \lfloor n/2 \rfloor$.

The roles of the remaining inputs, $\Lambda$, $\mu$, and $B$, are identical to their roles in IPSS. Briefly, $\Lambda = [\lambda_{\min}, \lambda_{\max}]$ is discretized to a finite grid with $\lambda_{\max}$ chosen large enough that no variables are selected, and the integral in \eqref{eq:ipss} is approximated by a Riemann sum. For $\mu$, we use the family of probability measures $\mu_\delta(d\lambda) \propto \lambda^{-\delta}\,d\lambda$; results are robust to the choice of $\delta$ \citep{melikechi2026integrated, melikechi2025nonparametric}. Similarly, results vary little with $B$ provided $B \geq 50$ \citep{shah2013variable, melikechi2026integrated}. Our default is $B=100$. 

\section{Theoretical Results}\label{sec:theory}


The central challenge of effect modifier discovery is that $\tau$ must be estimated, introducing additional uncertainty that can bias selections and inflate false discoveries. In this section, we characterize the performance of \css{} in terms of CATE estimation error, establishing a direct connection between causal variable selection and prediction. We formalize this correspondence in three steps. First, we prove that the selection probabilities estimated via \cref{alg:selection_probabilities} are asymptotically unbiased for their oracle counterparts (\cref{thrm:bias}). Second, we prove that the \css{} estimator $\hat{\mathcal{E}}_\gamma$ converges to the oracle selection set under CATE consistency and suitable subsampling conditions (\cref{thrm:convergence}). Third, we establish finite-sample false discovery control for \css{} (\cref{thrm:fdr}). Proofs are deferred to \cref{sup_sec:proofs}.

\subsection{Oracle selection}\label{sec:theory_bias}

\cref{thrm:bias,thrm:convergence} assume that $\hat{S}_\lambda$ selects variables by thresholding importance scores (\cref{sec:parameters}). Specifically, we consider selectors of the form
\begin{align}\label{eq:thresholding}
	\hat{S}_\lambda(\bX_A, f)
    =
    \bigl\{
        j : \Phi_j(\bX_A, f) \geq \lambda
    \bigr\},
\end{align}
where $\Phi_j : \mathbb{R}^{|A| \times p} \times \mathcal{F} \to \mathbb{R}$ quantifies the importance of $X_j$ using samples in $A$, and $\mathcal{F}$ is a normed space of functions from $\mathbb{R}^p$ to $\mathbb{R}$. In the traditional supervised setting, for example, $\Phi_j(\bX_A,f)=\Phi_j(\bX_A,f(\bX_A))$ is trained on data $\{(\bx_i,y_i):i\in A\}$, where the outcomes $y_i=f(\bx_i) + \varepsilon$ are directly observed and $\varepsilon$ is additive noise. In our setting, $f=\tau$. Common examples of importance scores include magnitudes of estimated regression coefficients obtained from fitting generalized linear models, the aforementioned mean decrease impurity (MDI) scores computed from tree-ensemble methods such as random forests or gradient boosting, and even more general scores such as permutation importance and SHAP values \citep{breiman2001,tibshirani1996regression}.

In the following, we fix $\tau$ and the subsample size $m$, and let $A\subseteq [n]$ be of size $m$. \cref{ass:lipschitz,ass:density} are used in \cref{thrm:bias,thrm:convergence,thrm:oracle_fdr}; they are not required for \cref{thrm:fdr}.

\begin{assumption}[Lipschitz importance]\label{ass:lipschitz}
There exists $L > 0$ such that, for all $\eta \in \mathcal{F}$,
\begin{align}\label{eq:lipschitz_importance}
    \bigl|
        \Phi_j(\bX_A, \eta)
        -
        \Phi_j(\bX_A, \tau)
    \bigr|
    \leq
    L \lVert\eta - \tau\rVert
    \quad
    \text{almost surely}.
\end{align}
\end{assumption}

\begin{assumption}[Bounded density]\label{ass:density}
The conditional random variable $\Phi_j(\bX_A, \tau) \mid \hat{\tau}_{A^c}$ admits a density that is bounded above by $M < \infty$.
\end{assumption}

\cref{ass:lipschitz} ensures that the score used by the base selector changes continuously when the CATE is perturbed. \cref{ass:density} rules out degenerate behavior at the importance score threshold, $\lambda$.

\begin{theorem}[Asymptotic unbiasedness]\label{thrm:bias}
Under \cref{ass:lipschitz,ass:density},
\begin{align}\label{eq:bias_bound}
    \bigl|
        \mathbb{E}\left(\hat{\pi}_{j,m}(\lambda)\right)
        -
        \pi_{j,m}(\lambda)
    \bigr|
    \leq
    2ML\,\mathbb{E}\lVert\hat{\tau}_{A^c} - \tau\rVert.
\end{align}
\end{theorem}

\cref{thrm:bias} shows that the bias in the estimated selection probabilities is controlled by the CATE estimation error, $\mathbb{E}\lVert\hat{\tau}_{A^c} - \tau\rVert$. The key point is that in \cref{alg:selection_probabilities}, the estimator $\hat{\tau}_{A^c}$ is trained on observations outside $A$, whereas $\Phi_j(\bX_A,\tau)$ depends only on observations in $A$. Conditional on the random split, these two quantities are independent. Therefore, disagreement between the estimated and oracle selectors can occur only when the oracle score $\Phi_j(\bX_A,\tau)$ lies within $L\lVert\hat{\tau}_{A^c}-\tau\rVert$ of the importance threshold $\lambda$. \cref{ass:density} controls the probability of this near-threshold event.

\begin{corollary}[Rate under CATE consistency]\label{cor:cate_rate}
Under the conditions of \cref{thrm:bias}, if $\mathbb{E}\lVert\hat{\tau}_{A^c}-\tau\rVert = o(1)$ as $n \to \infty$, then $\hat{\pi}_{j,m}(\lambda)$ is asymptotically unbiased for $\pi_{j,m}(\lambda)$. If, in addition, $\mathbb{E}\lVert\hat{\tau}_{A^c}-\tau\rVert^2 = O(n^{-1})$, then Jensen's inequality gives $\mathbb{E}\lVert\hat{\tau}_{A^c}-\tau\rVert = O(n^{-1/2})$, and the bias is $O(n^{-1/2})$.
\end{corollary}

We next show that the \css{} estimator $\hat{\mathcal{E}}_\gamma$ converges to its oracle counterpart,
\begin{align}\label{eq:oracle_selection_set}
    \mathcal{E}_\gamma^*
    =
    \left\{
        j :
        \int_\Lambda
        f\bigl(\pi_{j,m}(\lambda)\bigr)
        \,\mu(d\lambda)
        \geq
        \gamma
    \right\}.
\end{align}

The proof uses \cref{thrm:bias} together with a variance bound on the oracle version of $\hat{\pi}_{j,m}$, in which $\tau$ replaces $\hat{\tau}_{A^c}$; see \cref{lem:variance} in \cref{sup_sec:proofs}.

\begin{theorem}[Oracle selection]\label{thrm:convergence}
Suppose $f$ in \eqref{eq:ipss} is $L_f$-Lipschitz on $[0,1]$ and define
\begin{align}\label{eq:margin_condition}
    \varepsilon_j
    =
    \left|
        \int_\Lambda
        f\bigl(\pi_{j,m}(\lambda)\bigr)
        \,\mu(d\lambda)
        -
        \gamma
    \right|
    >
    0
    \quad
    \text{for all } j.
\end{align}
Set $\varepsilon = \min_j \varepsilon_j$. Under \cref{ass:lipschitz,ass:density},
\begin{align}\label{eq:convergence_bound}
    \mathbb{E}
    \bigl|
        \hat{\mathcal{E}}_\gamma
        \,\Delta\,
        \mathcal{E}_\gamma^*
    \bigr|
    \leq
    \frac{L_f p}{\varepsilon}
    \left(
        2ML\,\mathbb{E}\lVert\hat{\tau}_{A^c}-\tau\rVert
        +
        \sqrt{
            \frac{1}{4B}
            +
            \frac{2m^2}{n}
            +
            4ML\,\mathbb{E}\lVert\hat{\tau}_{A^c}-\tau\rVert
        }
    \right).
\end{align}
\end{theorem}

The margin condition \eqref{eq:margin_condition} excludes the degenerate case in which variables lie exactly on the selection boundary. Since the integrated oracle scores $\smallint f(\pi_{j,m}(\lambda))\, \mu(d\lambda)$ are $p$ deterministic numbers, this condition holds for at most $p$ values of $\gamma$ (and hence Lebesgue almost-everywhere on $\Lambda$). The bound \eqref{eq:convergence_bound} separates the three sources of error: Monte Carlo error from finite $B$, dependence from overlapping subsamples (the $2m^2/n$ term), and CATE estimation error.

\begin{corollary}[Oracle convergence]\label{cor:oracle_convergence}
Under the conditions of \cref{thrm:convergence}, if $B \to \infty$, $m = o(\sqrt{n})$, and $\mathbb{E}\lVert\hat{\tau}_{A^c}-\tau\rVert \to 0$, then $\mathbb{E}\lvert\hat{\mathcal{E}}_\gamma\,\Delta\,\mathcal{E}_\gamma^*\rvert\to 0$
\end{corollary}

\subsection{False discovery control}\label{sec:theory_fdr}

We now establish a finite-sample upper bound on the expected number of false positives, $\EFP=\mathbb{E}\lvert\hat{\mathcal{E}}_\gamma \cap \mathcal{E}^c\rvert$, selected by \css{}. Let
\begin{align}\label{eq:qhat_def}
    \hat{q}(\lambda)
    =
    \mathbb{E}
    \bigl|
        \hat{S}_\lambda(\bX_A,\hat{\tau}_{A^c})
    \bigr|
\end{align}
denote the expected number of variables selected by $\hat{S}_\lambda$ on samples $\{(\bX_i,\hat{\tau}_{A^c}(\bX_i)) : i\in A\}$, where $\hat{\tau}_{A^c}$ is trained on samples in $A^c$. Similarly, let $q(\lambda)=\mathbb{E}\lvert\hat{S}_\lambda(\bX_A,\tau)\rvert$ be the expected number of variables selected by $\hat{S}_\lambda$ on samples in $A$, but with the true CATE in place of $\hat{\tau}_{A^c}$.
\begin{assumption}\label{ass:exchangeability}
For $r\in\mathbb{N}$, let $A_1,\dots,A_{2r}$ be subsets of $[n]$ of size $m$, drawn as in \cref{alg:selection_probabilities}. We say that \cref{ass:exchangeability}a holds for $r$ if
\begin{align}\label{eq:exchangeability_condition_a}
    \max_{j \in \mathcal{E}^c}
    \int_\Lambda
    \mathbb{P}\!\left(
        j \in
        \bigcap_{b=1}^{2r}
        \hat{S}_\lambda(\bX_{A_b},\hat{\tau}_{A_b^c})
    \right)
    \mu(d\lambda)
    \leq
    \int_\Lambda
    \left(
        \frac{\hat{q}(\lambda)}{p}
    \right)^{2r}
    \mu(d\lambda).
\end{align}
Similarly, we say that \cref{ass:exchangeability}b holds for $r$ if
\begin{align}\label{eq:exchangeability_condition_b}
    \max_{j \in \mathcal{E}^c}
    \int_\Lambda
    \mathbb{P}\!\left(
        j \in
        \bigcap_{b=1}^{2r}
        \hat{S}_\lambda(\bX_{A_b},\tau)
    \right)
    \mu(d\lambda)
    \leq
    \int_\Lambda
    \left(
        \frac{q(\lambda)}{p}
    \right)^{2r}
    \mu(d\lambda).
\end{align}
\end{assumption}

Assumptions \ref{ass:exchangeability}a and \ref{ass:exchangeability}b say that, on average over all $\lambda\in\Lambda$, the probability that $\hat{S}_\lambda$ selects the same null variable across multiple random subsamples of the data is no larger than the corresponding average probability under uniformly random selection with expected size $\hat{q}(\lambda)$ or $q(\lambda)$, respectively. The only difference between the two is that \cref{ass:exchangeability}a is stated in terms of the estimated CATEs, $\hat{\tau}_{A_b^c}$, while \cref{ass:exchangeability}b uses the true CATE, $\tau$.

\begin{remark}[Relaxation of the IPSS condition]\label{rem:weakened_condition}
\Cref{ass:exchangeability} is stated in terms of integrals over $\Lambda$, whereas the analogous condition in \citet{melikechi2026integrated} (called Condition 1 in that paper) is required to hold pointwise in $\lambda$. Since the pointwise inequality implies the integrated one, \cref{ass:exchangeability} is strictly weaker. This same relaxation applies to the original IPSS false discovery bound and requires no changes to its proof.
\end{remark}

\begin{theorem}[False discovery control]\label{thrm:fdr}
Suppose \cref{ass:exchangeability}a holds for $r \in \{1,2,3\}$. Then
\begin{align}\label{eq:fdr_bound}
    \mathbb{E}
    \bigl|
        \hat{\mathcal{E}}_\gamma
        \cap
        \mathcal{E}^c
    \bigr|
    \leq
    \frac{1}{\gamma}
    \int_\Lambda
        \frac{\hat{q}(\lambda)^2}{B^2 p}
        +
        \frac{3(B-1)\hat{q}(\lambda)^4}{B^2 p^3}
        +
        \frac{(B-1)(B-2)\hat{q}(\lambda)^6}{B^2 p^5}\,
    \mu(d\lambda).
\end{align}
\end{theorem}

The proof follows the structure of the IPSS $\EFP$ bound with two modifications: $\hat{q}(\lambda)$ is defined for the cross-fitted causal selector $\hat{S}_\lambda(\bX_A,\hat{\tau}_{A^c})$, and the multinomial expansion at the heart of the argument requires only the integrated form of \cref{ass:exchangeability} (\cref{rem:weakened_condition}). \cref{thrm:fdr} also holds for arbitrary base selectors $\hat{S}_\lambda$, since its proof does not invoke \cref{ass:lipschitz,ass:density}. The next result, \cref{thrm:oracle_fdr}, provides a complementary bound under the oracle exchangeability condition, \cref{ass:exchangeability}b. The bound matches the IPSS false-positive bound with $q(\lambda)$ in place of $\hat{q}(\lambda)$, plus an additional term $\rho$ from \cref{thrm:convergence} that vanishes asymptotically.

\begin{theorem}[Oracle false discovery control]\label{thrm:oracle_fdr}
If the assumptions of \cref{thrm:convergence} hold, and if
\cref{ass:exchangeability}b holds for $r \in \{1,2,3\}$, then
\begin{align}\label{eq:oracle_fdr_bound}
    \mathbb{E}
    \bigl|
        \hat{\mathcal{E}}_\gamma
        \cap
        \mathcal{E}^c
    \bigr|
    \leq
    \rho
    +
    \frac{1}{\gamma}
    \int_\Lambda
        \frac{q(\lambda)^2}{B^2 p}
        +
        \frac{3(B-1)q(\lambda)^4}{B^2 p^3}
        +
        \frac{(B-1)(B-2)q(\lambda)^6}{B^2 p^5}\,
    \mu(d\lambda),
\end{align}
where $\rho$ denotes the right-hand side of \eqref{eq:convergence_bound}.
\end{theorem}

The quantity $\rho$ in \eqref{eq:oracle_fdr_bound} controls the discrepancy between the cross-fitted and oracle selection sets. By \cref{cor:oracle_convergence}, if $B \to \infty$, $m = o(\sqrt{n})$, and $\mathbb{E}\lVert\hat{\tau}_{A^c}-\tau\rVert \to 0$, then $\rho \to 0$ and the bound reduces to the IPSS false-positive bound; that is, the expected number of false positives in the \css{} estimator $\hat{\mathcal{E}}_\gamma$ converges to the expected number of false positives in the oracle selection set $\mathcal{E}^*_\gamma$.

\begin{remark}[FDR control via efp scores]\label{rem:fdr}
\Cref{thrm:fdr} bounds $\EFP = \mathbb{E}|\hat{\mathcal{E}}_\gamma \cap \mathcal{E}^c|$ directly. In practice, we convert this bound into an \textit{efp (expected false positive) score} $\efp(j)$ for each variable $j$, defined as the smallest value of $\EFP$ at which $j$ would be selected \citep{melikechi2025nonparametric}. Given a target $\EFP$ level $t$, the selection set $\{j : \efp(j) \leq t\}$ contains at most $t$ false positives on average. For FDR control at level $\alpha$, we choose the largest $t$ for which the resulting selection set $\hat{\mathcal{E}}(t) = \{j : \efp(j) \leq t\}$ satisfies $t / |\hat{\mathcal{E}}(t)| \leq \alpha$, which follows from the approximation $\mathrm{FDR} \approx \EFP / \mathbb{E}|\hat{\mathcal{E}}(t)|$ \citep{storey2003positive}.
\end{remark}


\section{Simulation Studies}\label{sec:simulations}


We conduct a series of simulation studies to evaluate the performance of \css{} relative to other effect modifier discovery methods across a range of settings. There are 34 experiments total, 17 in which the prognostic function and CATE are linear functions of the covariates, and 17 in which they are nonlinear (\cref{tab:sim_params}). Each experiment consists of 200 trials. In \cref{sec:sim_design}, we describe the simulation designs. In \cref{sec:sim_methods}, we describe the evaluation protocol and method implementations. In \cref{sec:sim_results}, we present the results. Additional results are reported in \cref{sup_sec:simulations}.

\subsection{Data generation}\label{sec:sim_design}

We generate data from the potential outcomes model~\eqref{eq:decomposition} with an additional parameter, $a$, that controls the relative contributions of the prognostic variables and effect modifiers to the outcome; see \eqref{eq:prog_strength} below. Covariates $\bX \in \mathbb{R}^p$ are drawn from a multivariate normal distribution with mean zero and Toeplitz covariance $\bm{\Sigma} \in \mathbb{R}^{p \times p}$, $\Sigma_{ij} = \rho^{|i-j|}$, then standardized to unit empirical variance. The sets $\mathcal{P}$ and $\mathcal{E}$ are subsets of $[p]$ of fixed sizes $|\mathcal{P}|$ and $|\mathcal{E}|$, drawn uniformly without replacement and independently of each other; $\mathcal{P} \cap \mathcal{E}$ may therefore be nonempty. Treatment assignment is modeled after a randomized control trial, $Z_i \sim \mathrm{Bernoulli}(1/2)$, in 30 of the 34 experiments. In the other four, there are 5 or 10 confounders drawn uniformly at random from $[p]$, which affect treatment assignment via a logistic propensity score, $\Pr(Z_i = 1 \mid \bX_i) = \bigl(1 + \exp({-\textstyle\sum_{j \in \mathcal{C}} X_{ij}})\bigr)^{-1}$, where $\mathcal{C} \subseteq [p]$ indexes the confounders. Confounders may overlap with $\mathcal{P}$ and $\mathcal{E}$.

In the linear case, the prognostic function and CATE are given by
\begin{align*}
  \mu(\bX) = \bX^\top\bbeta \quad\text{and}\quad \tau(\bX) = \bX^\top\bgamma,
\end{align*}
where the nonzero coordinates of $\bbeta$ and $\bgamma$ are indexed by $\mathcal{P}$ and $\mathcal{E}$, respectively, and drawn independently and uniformly from $[-1, -0.5] \cup [0.5, 1]$. In the nonlinear setting,
\begin{align*}
  \mu(\bX) = \sum_{j \in \mathcal{P}} \exp\!\left(-\tfrac{1}{8}X_j^2\right) \quad\text{and}\quad \tau(\bX) = \sum_{j \in \mathcal{E}} \exp\!\left(-\tfrac{1}{8}X_j^2\right).
\end{align*}
In both settings, $\mu$ and $\tau$ are centered and standardized to unit empirical variance within each simulated dataset. The observed outcome is
\begin{align}\label{eq:prog_strength}
  Y = a\mu(\bX) + Z\tau(\bX) + \varepsilon,
\end{align}
where $\varepsilon \sim \mathcal{N}(0, \sigma^2)$ and $\sigma^2$ is calibrated to achieve a prespecified signal-to-noise ratio,
\begin{align*}
  \mathrm{SNR} = \var\bigl(a\mu(\bX) + Z\tau(\bX)\bigr) \,/\, \sigma^2,
\end{align*}
with the variance computed empirically on each simulated dataset. The parameter $a \geq 0$ controls the relative magnitude of the prognostic and treatment components: in the RCT case, $\var(a\mu(\bX)) = a^2$ and $\var(Z\tau(\bX)) = 1/2$, so $a = 1/\sqrt{2}$ yields equal contributions, while the default $a = 1$ makes the prognostic component twice as variable as the treatment component. Parameter values for each experiment are summarized in \cref{tab:sim_params}.

\begin{table}[h]
\centering
\begin{tabular}{l | ccc | ccc}
\toprule
\multicolumn{1}{l}{Parameter} & \multicolumn{3}{c}{\hspace*{-0.5em}Linear} & \multicolumn{3}{c}{Nonlinear} \\
\midrule
$n$             & $500$ & $\mathbf{1000}$ & $2000$ & $500$ & $\mathbf{1000}$ & $2000$ \\
$p$             & $50$  & $\mathbf{100}$  & $200$  & $25$  & $\mathbf{50}$   & $100$  \\
$\lvert\mathcal{E}\rvert$ & $5$   & $\mathbf{10}$   & $15$   & $2$   & $\mathbf{5}$    & $10$   \\
$\lvert\mathcal{P}\rvert$ & $5$   & $\mathbf{10}$   & $15$   & $2$   & $\mathbf{5}$    & $10$   \\
$\lvert\mathcal{C}\rvert$     & $\mathbf{0}$ & $5$ & $10$ & $\mathbf{0}$ & $5$ & $10$ \\
$\rho$          & $0$   & $\mathbf{0.5}$  & $0.75$ & $0$   & $\mathbf{0.5}$  & $0.75$ \\
$\mathrm{SNR}$  & $0.5$ & $\mathbf{1}$    & $2$    & $0.5$ & $\mathbf{1}$    & $2$    \\
$a$             & $\tfrac{1}{\sqrt{2}}$ & $\mathbf{1}$ & $1.5$ & $\tfrac{1}{\sqrt{2}}$ & $\mathbf{1}$ & $1.5$ \\
\bottomrule
\end{tabular}
\caption{\textit{Simulation settings}. The baseline linear and nonlinear experiments correspond to all parameters set to their default values, shown in bold. The 32 single-parameter perturbations (16 per setting) each vary one parameter from its default to one of the non-bold values, holding the others fixed. Symbols: $n$ (sample size); $p$ (number of covariates); $\lvert\mathcal{E}\rvert$ (number of effect modifiers); $\lvert\mathcal{P}\rvert$ (number of prognostic variables); $\lvert\mathcal{C}\rvert$ (number of confounders); $\rho$ (correlation strength); $\mathrm{SNR}$ (signal-to-noise ratio); $a$ (relative prognostic strength).}
\label{tab:sim_params}
\end{table}

\subsection{Methods}\label{sec:sim_methods}

We compare \css{} to six methods: IPSS \citep{melikechi2026integrated}, knockoffs \citep{barber2015controlling}, the Benjamini-Hochberg procedure \citep{benjamini1995controlling}, Lasso \citep{tibshirani1996regression}, and variable importance scores from XGBoost \citep{chen2016xgboost} and causal forests \citep{wager2018estimation}. \cref{tab:methods} summarizes the key implementation details, which we elaborate upon below.

All methods except \css{} and causal forest operate on cross-fitted pseudo-outcomes $\hat{\tau}_i$ that serve as individual-level estimates of $\tau(\bx_i)$. These are generated via $K$-fold cross-fitting as described in \cref{sec:intro}: for each fold $k$, an X-Learner is fit on the $K-1$ training folds and applied to the held-out fold to obtain $\hat{\tau}_i$ for each $i \in A_k$. We use ridge regression for the X-Learner's base models in the linear setting and gradient boosting in the nonlinear setting. We use $K = 5$ folds for all methods except XGB-VI, for which $K = 10$ showed improved performance compared to $K=5$; the other methods showed no improvement at $K = 10$.

\begin{table}[h]
\centering
\begin{tabular}{l c c l l}
\toprule
Method & FDR control & Folds & Base method & Software \\
\midrule
\css{}            & \checkmark & ---  & Lasso/XGBoost                                           & ---                   \\
IPSS              & \checkmark & $5$  & Lasso/XGBoost                                           & \texttt{ipss}         \\
KO                & \checkmark & $5$  & Lasso/GLM      & \texttt{knockoff} (R) \\
BH                & \checkmark & $5$  & OLS $p$-values                                          & \texttt{statsmodels}  \\
Lasso             & $\times$   & $5$  & cross-validated $\ell_1$                                & \texttt{scikit-learn} \\
XGB-VI            & $\times$   & $10$ & XGBoost gain                                            & \texttt{xgboost}      \\
CF-VI             & $\times$   & ---  & intrinsic VI                                            & \texttt{econml}       \\
\bottomrule
\end{tabular}
\caption{Methods compared in the simulation studies. The Folds column denotes the number of cross-fitting folds used to generate pseudo-outcomes $\hat{\tau}_i$; \css{} and CF-VI operate directly on $(\bX, Y, Z)$ and perform their own internal cross-fitting. For \css{}, IPSS, and KO, slashed entries indicate linear/nonlinear configurations; for KO, \texttt{lasso\_coefdiff} and \texttt{glmnet\_coefdiff} are the corresponding knockoff statistics from the \texttt{knockoff} R package. CF-VI uses the variable importance scores intrinsic to the causal forest. All methods are implemented using Python packages except \texttt{knockoff}.}
\label{tab:methods}
\end{table}

\css{}, IPSS, KO, and BH each target a nominal FDR level $\alpha \in [0, 0.5]$. The remaining three methods yield a fixed selection that does not vary with a nominal level: Lasso selects all covariates with nonzero coefficients at the cross-validation-selected penalty $\lambda$, while XGB-VI and CF-VI rank covariates by their respective importance scores and select the top $|\mathcal{E}|$. The top-$|\mathcal{E}|$ rule provides these two methods with the oracle benefit of knowing the true number of effect modifiers, which they would not have in practice; their reported performance therefore represents an upper bound on what is achievable from the underlying rankings. As a direct consequence of selecting exactly $|\mathcal{E}|$ variables, their TPR and FDR are mechanically related by $\mathrm{FDR} = 1 - \mathrm{TPR}$.

Performance is summarized by the true positive rate (TPR) and false discovery rate (FDR), defined as in \cref{sec:intro}, averaged across the 200 trials. For \css{}, IPSS, KO, and BH, we plot mean TPR and mean FDR as functions of $\alpha$. The diagonal dashed line in the FDR panels serves as a reference for perfect nominal control; a method achieves control if its empirical FDR curve lies on or below this line. For Lasso, XGB-VI, and CF-VI, mean TPR and mean FDR are displayed as horizontal lines.

\subsection{Results}\label{sec:sim_results}

\begin{figure}
\includegraphics[width=\textwidth]{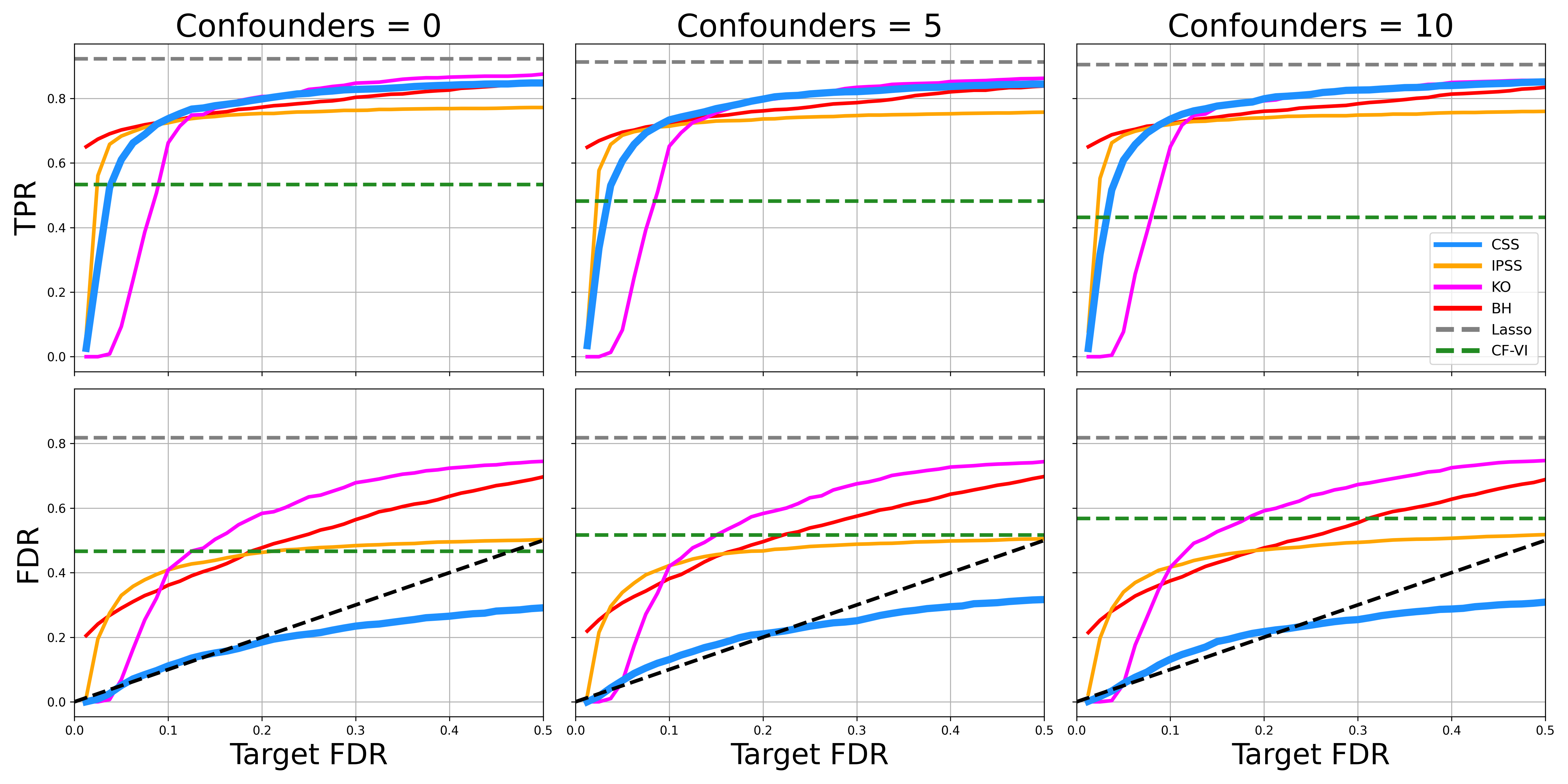}%
\caption{\textit{Linear results: Confounding variables.} Top and bottom rows show mean TPR and mean FDR (averaged over 200 trials); the diagonal dashed line indicates perfect nominal FDR control. \textbf{Left:} $|\mathcal{C}| = 0$ (RCT). \textbf{Middle:} $|\mathcal{C}| = 5$. \textbf{Right:} $|\mathcal{C}| = 10$.}
\label{fig:linear_confounders}
\end{figure}

\begin{figure}
\includegraphics[width=\textwidth]{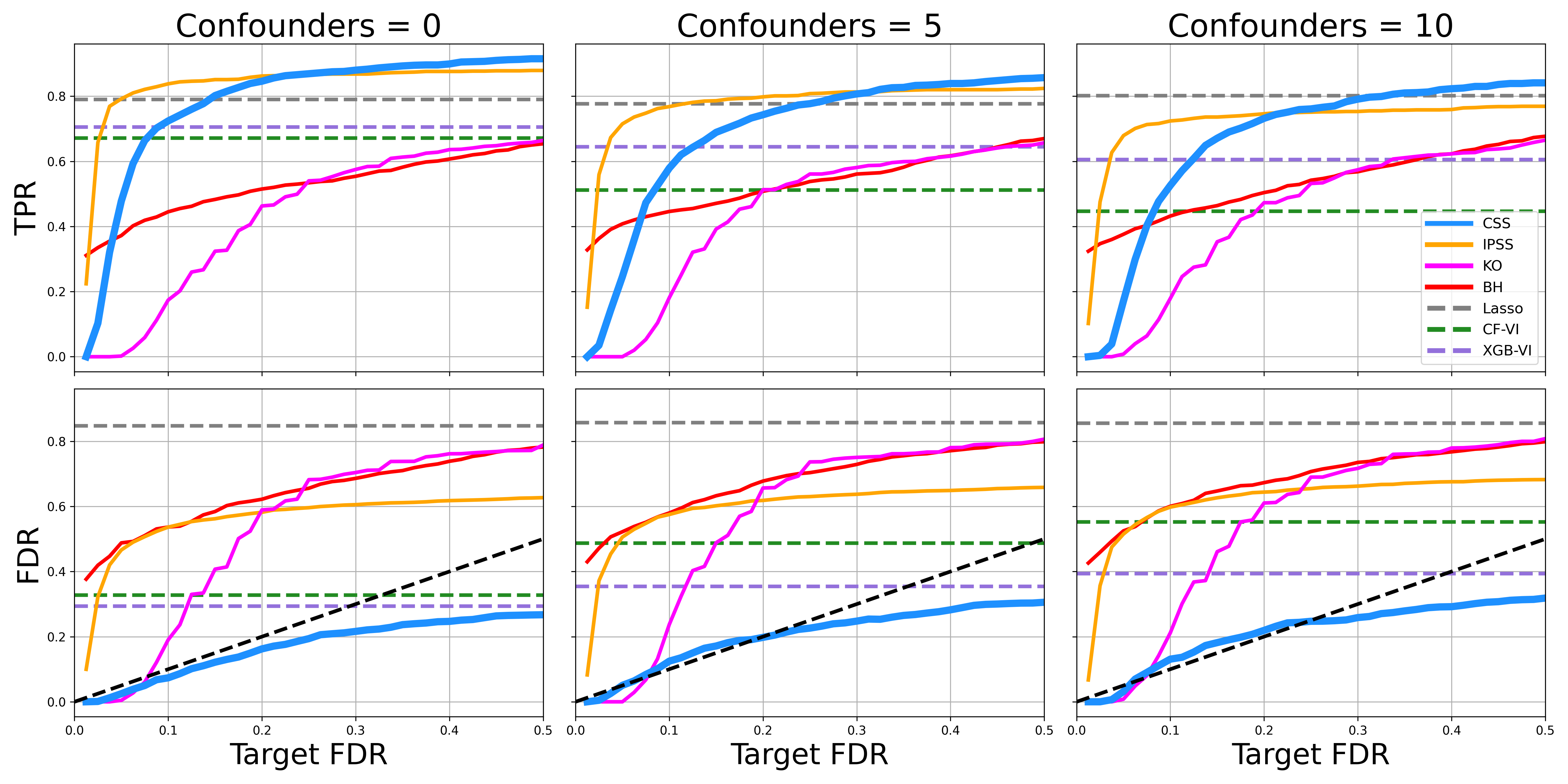}%
\caption{\textit{Nonlinear results: Confounding variables.} As in \cref{fig:linear_confounders} but in the nonlinear setting.}
\label{fig:nonlinear_confounders}
\end{figure}

\cref{fig:linear_confounders,fig:nonlinear_confounders} report results in the linear and nonlinear settings, respectively, as the number of confounders varies from 0 (RCT) to 10. The headline finding is consistent across both settings: \css{} is the only method that maintains empirical FDR at or below the nominal level $\alpha$ uniformly across confounding regimes, and it does so without sacrificing power relative to the other FDR-controlling alternatives. 

\cref{fig:linear_confounders} shows results in the linear setting. The empirical FDR of \css{} lies on or below the diagonal across the full range of target levels and is essentially unchanged as the number of confounders increases from 0 to 10. The other FDR-controlling methods are anti-conservative across all three regimes: IPSS plateaus at empirical FDR near 0.5, KO rises above 0.7, and BH exceeds 0.5 once $\alpha \geq 0.1$. The non-FDR-controlling baselines, Lasso and CF-VI, produce empirical FDR around 0.8 and 0.5, respectively. Despite its conservatism in FDR, \css{} attains TPR comparable to the best alternatives, reaching approximately 0.85 by $\alpha = 0.1$ and matching BH, KO, and Lasso in the high-$\alpha$ regime while substantially exceeding CF-VI and the plateaued IPSS curve.

\cref{fig:nonlinear_confounders} shows the corresponding nonlinear results. The qualitative pattern is preserved: \css{} is again the only method that controls FDR uniformly across confounding regimes. \css{} attains the highest TPR among methods with valid FDR control: by $\alpha = 0.2$ its TPR exceeds 0.85, surpassing the other FDR-controlling methods. The non-FDR-controlling baselines (Lasso, XGB-VI, CF-VI) achieve TPR between 0.5 and 0.8 but at empirical FDR levels of 0.3 to 0.85. The improvement of \css{} over other methods, such as IPSS, in the nonlinear setting is consistent with the difficulty of CATE estimation being more pronounced when the underlying response surface is nonlinear, which the cross-fitted design of \css{} is explicitly intended to address.

Additional experiments varying sample size $n$, dimension $p$, number of effect modifiers $|\mathcal{E}|$, number of prognostic variables $|\mathcal{P}|$, covariate correlation $\rho$, signal-to-noise ratio, and prognostic strength $a$ are reported in \cref{sup_sec:simulations}. The patterns observed in \cref{fig:linear_confounders,fig:nonlinear_confounders} are broadly consistent with those additional results.

\section{Application}\label{sec:realdata}

We apply \texttt{CausalStabSel} to two real-world datasets: a randomized controlled trial (RCT) and an observational study. The RCT data come from Amgen's PRIME trial, which evaluates the addition of a monoclonal-antibody therapy to standard chemotherapy in patients with previously untreated metastatic colorectal cancer (mCRC) \citep{douillard2010randomized}. The observational data come from the Pennsylvania birth-records study of maternal smoking and infant birthweight, a publicly available subset of the data used by \citet{cattaneo2010}. The objective in both settings is to assess the behavior of 
\texttt{CausalStabSel} on real data with partially known structure.

\subsection{PRIME Randomized Trial}

PRIME is a multicenter, open-label, randomized phase~3 trial comparing \textit{panitumumab}---a fully human anti-EGFR monoclonal antibody---plus FOLFOX4 chemotherapy to FOLFOX4 alone in mCRC \citep{douillard2010randomized}. The trial demonstrated a PFS benefit confined to KRAS wild-type patients, later refined by extended-RAS analyses to patients without mutations in KRAS or NRAS exons 2--4 \citep{douillard2013panitumumab}. Patient-level data were obtained through Project Data Sphere \citep{green2015project}.

The analytic sample consists of $N = 862$ patients with $p = 30$ baseline covariates spanning demographics, clinical features, mutation biomarkers, and laboratory values; the outcome is progression-free survival at 365 days, treated as a binary landmark. Several covariates are highly correlated (Figure~\ref{fig:corr_309}). The data are split into a discovery sample ($N = 603$, 70\%) and a held-out inference sample ($N = 259$, 30\%); an 80/20 split is used as a sensitivity analysis. As a baseline, plain AIPW regression on the discovery sample finds 2--4 covariates significant at $p < 0.05$ univariately and 6--7 multivariately (Supplementary Material Section~\ref{sup_sec:application}).
\subsubsection{Causal Stability Selection Results.}

We run \texttt{CausalStabSel} with: (i) DR-learner pseudo-outcomes with Ridge nuisance models and 5-fold cross-fitting; (ii) LASSO as the main selector and Gradient Boosting (GB) as a sensitivity check; (iii) $B = 500$ subsamples, $\delta = 2$, subsample size $n_{\text{disc}}/4$, and a $q$-value threshold of $0.10$. \texttt{CausalStabSel} selects five effect modifiers: platelets, BMMTR1 (KRAS exon 2), BMMTR3 (KRAS exon 4), age, and lactate dehydrogenase. The four-feature core (BMMTR1, BMMTR3, age, platelets) is selected in all four configurations of split ratio and selector (Supplementary Material Section~\ref{sup_sec:application}). Recovering BMMTR1 (KRAS exon 2) and BMMTR3 (KRAS exon 4) provides a calibration check: KRAS mutations are the canonical predictors of resistance to anti-EGFR therapy and motivated the trial's original subgroup analysis \citep{douillard2010randomized, douillard2013panitumumab}. The three additional selections---platelets, age, and lactate dehydrogenase---have established but less canonical associations with anti-EGFR response: elevated platelet counts mark systemic inflammation and angiogenesis pathways partially independent of EGFR signaling, while age and LDH are recognized prognostic factors with documented associations with treatment response. The goal of this application, however, is to assess the method's behavior, not to advance new clinical claims.

For corroboration, we run held-out inference using the BLP-on-CATE specification of \citet{chernozhukov2025generic}: regressing the estimated CATE on each de-meaned modifier individually, with HC1 robust standard errors. The estimated ATE is essentially zero ($-0.016$, 95\% CI $[-0.144, 0.112]$), consistent with the trial's well-known finding that benefit is confined to the KRAS wild-type subgroup \citep{douillard2010randomized}. Higher baseline platelet counts are associated with a smaller treatment effect (coefficient $-0.0014$, 95\% CI $[-0.0024, -0.0004]$); the other selected features do not reach significance in the held-out inference. This may reflect limited statistical power in the small held-out sample ($n_{\text{inf}} = 259$), especially for KRAS exon-4 mutations (prevalent in $<5\%$ of patients), but the possibility that some selections do not replicate cannot be ruled out.


\subsection{Birthweight Observational Study}

The observational analysis uses data from \citet{cattaneo2010}, a study of the effect of maternal smoking on infant birthweight drawn from Pennsylvania vital statistics birth records. The dataset has become a standard benchmark for evaluating causal-inference methods in observational settings \citep{di2026aggregation}. The data used here are a publicly available random sample of the original.

The analytic sample consists of $N = 4{,}642$ observations with $p = 15$ baseline covariates spanning maternal and paternal demographics, parity, marital status, and prenatal behaviors; the outcome is infant birthweight in grams (continuous), and the treatment is maternal smoking during pregnancy, with prevalence $18.6\%$. Several covariates are highly correlated (Figure~\ref{fig:corr_bw}), notably between maternal and paternal characteristics under assortative mating (Mother White / Father White, $r = 0.85$) and between mechanically related parity variables (Birth order / First baby, $r = -0.70$). The data are split into a discovery sample ($N = 3{,}249$, 70\%) and a held-out inference sample ($N = 1{,}393$, 30\%); an 80/20 split is used as a sensitivity analysis. As a baseline, plain AIPW regression on the discovery sample finds 3--7 covariates significant at $p < 0.05$ univariately and zero significant covariates multivariately (Supplementary Material Section~\ref{sup_sec:application}).

\subsubsection{Causal Stability Selection Results.}

We run \texttt{CausalStabSel} with the same configuration as in the PRIME application: DR-learner pseudo-outcomes with Ridge nuisance models and 5-fold cross-fitting; LASSO as the selector; $B = 500$ subsamples, $\delta = 2$, subsample size $n_{\text{disc}}/4$, and a $q$-value threshold of $0.10$. \texttt{CausalStabSel} selects two effect modifiers: Mother age and First baby; both are selected in the 70/30 and 80/20 splits (Supplementary Material Section~\ref{sup_sec:application}). Sensitivity-analysis results using Gradient Boosting as the selector are also reported in the Supplementary Material. Both Mother age and First baby have established associations with smoking-related birthweight reduction in the epidemiological literature: older mothers tend to be exposed to lower cumulative smoking dose and to have better-developed uterine vasculature; nulliparous women have less developed uterine and placental vasculature than multiparous women, potentially compounding the harm from smoking-induced vasoconstriction. The goal of this application, however, is to assess the method's behavior, not to advance new clinical claims.

For corroboration, we run held-out inference using the BLP-on-CATE specification of \citet{chernozhukov2025generic}: regressing the estimated CATE on each de-meaned modifier individually, with HC1 robust standard errors. The estimated ATE is $-266$g (95\% CI $[-331, -202]$), consistent with the consensus epidemiological estimate of a $200$--$300$g smoking penalty and providing a calibration check on the AIPW pseudo-outcomes. In the 70/30 inference sample, Mother age (coefficient $-14.9$g per year, 95\% CI $[-26.3, -3.5]$) and First baby (coefficient $+145.8$g, 95\% CI $[17.0, 274.6]$) both reach significance. Neither modifier retains significance in the 80/20 split, likely reflecting limited power in the smaller held-out sample ($n_{\text{inf}} = 929$ vs.\ $1{,}393$). Treatment-effect heterogeneity is detected here against a strongly significant negative ATE, a regime distinct from PRIME and complementary as a test of the method's behavior.

\subsection{Application remarks} 

Together, these applications demonstrate \css{} on two complementary regimes. PRIME involves a near-null average treatment effect, where heterogeneity is the substantive finding and recovering the canonical KRAS modifier serves as calibration. The birthweight study involves a strongly significant negative ATE, where heterogeneity is detected against a clear treatment effect and the estimate of the ATE itself provides an independent calibration of the pseudo-outcomes. In both cases \css{} identifies a parsimonious set of effect modifiers with established literature support, and held-out BLP-on-CATE inference provides corroboration for at least a subset of selections.

\section{Discussion}\label{sec:discussion}


We introduced \css{}, a method for effect modifier discovery that combines cross-fitted CATE estimation with integrated path stability selection. The procedure accommodates arbitrary treatment effect estimators and arbitrary base selectors, and produces a selection set with an explicit, non-asymptotic upper bound on the expected number of false positives. Our theoretical results connect the bias of the estimated selection probabilities directly to the convergence rate of the underlying CATE estimator, establishing a quantitative link between treatment effect estimation and effect modifier discovery. Across a comprehensive simulation grid and two real-world applications spanning a randomized oncology trial and an observational study of maternal smoking, \css{} consistently controls false discovery while maintaining competitive power.


Several extensions follow naturally from the framework developed here. The selection probability machinery applies, with appropriate modification, to subgroup discovery, time-varying treatments, selection subspaces and settings with interference, where the relevant heterogeneity object differs from the CATE but the cross-fitted stability selection structure is preserved \citep{zhang2025quantifying}. On the inferential side, post-selection inference for the effect modifiers identified by \css{} (e.g., building on the held-out BLP-on-CATE strategy used in our applications, but with formal guarantees) would convert discovery into estimation with valid coverage. We leave these directions for future work.

\bibliographystyle{plainnat}
\bibliography{biblio}

\pagebreak

\appendix
\counterwithin{figure}{section}
\counterwithin{table}{section}

\renewcommand{\thepage}{S\arabic{page}}
\setcounter{page}{1}

\renewcommand{\thesection}{S\arabic{section}}

\renewcommand{\thefigure}{S\arabic{figure}}
\renewcommand{\thetable}{S\arabic{table}}

\begin{center}
    \textbf{\large Supplementary Material for ``Causal Stability Selection''}\\ \vspace{0.25cm}
    \normalsize Falco J. Bargagli-Stoffi, Omar Melikechi
\end{center}

\section{Proofs}\label{sup_sec:proofs}


\subsection{\cref{sec:theory_bias} proofs}\label{sup_sec:oracle_selection}

\begin{proof}[Proof of \cref{thrm:bias}]
Set $\Phi_j=\Phi_j(\bX_A,\tau)$ and $\hat{\Phi}_j=\Phi_j(\bX_A,\hat{\tau}_{A^c})$ and define the events
\begin{align*}
E_j(\lambda) &= \{j\in\hat{S}_\lambda(\bX_A,\tau(\bX_A))\} 
	= \{\Phi_j \geq \lambda\}; \\
\hat{E}_j(\lambda) &= \{j\in\hat{S}_\lambda(\bX_A,\hat{\tau}_{A^c}(\bX_A))\} 
	= \{\hat{\Phi}_j \geq \lambda\}.
\end{align*}
The following sequence of equalities and inequalities constitutes the full proof and is presented without interruption for clarity. Details about the individual steps are provided subsequently. In the following, $B_1 \,\Delta\, B_2 = (B_1\setminus B_2)\cup(B_2\setminus B_1)$ denotes the symmetric difference of sets $B_1$ and $B_2$.
\begin{align*}
\lvert\mathbb{E}\left(\hat{\pi}_j(\lambda)\right) - \pi_j(\lambda)\rvert &= \big\lvert\mathbb{P}(\hat{E}_j(\lambda)) - \mathbb{P}\left(E_j(\lambda)\right)\big\rvert \\
	&\leq \mathbb{P}\left(\hat{E}_j(\lambda) \,\Delta\, E_j(\lambda)\right) \\
	&\leq \mathbb{P}\left(\lvert\Phi_j - \lambda\rvert \leq \lvert\hat{\Phi}_j - \Phi_j\rvert\right) \\
	&\leq \mathbb{P}\left(\lvert\Phi_j - \lambda\rvert \leq L\lVert\hat{\tau}_{A^c} - \tau\rVert\right) \\
	&= \mathbb{E}\left[\mathbb{P}\left(\lvert\Phi_j - \lambda\rvert \leq L\lVert\hat{\tau}_{A^c} - \tau\rVert\right) \mid \hat{\tau}_{A^c}\right] \\
	&\leq \mathbb{E}\left(2ML\lVert\hat{\tau}_{A^c} - \tau\rVert\right).
\end{align*}
The first equality is immediate. The first inequality holds because, for any sets $B_1$ and $B_2$,
\begin{align*}
\lvert\mathbb{P}(B_1) - \mathbb{P}(B_2)\rvert &= \big\lvert\left[\mathbb{P}(B_1\setminus B_2) + \mathbb{P}(B_1\cap B_2)\right] - \left[\mathbb{P}(B_2\setminus B_1) + \mathbb{P}(B_2\cap B_1)\right]\big\rvert \\
	&= \lvert\mathbb{P}(B_1\setminus B_2) - \mathbb{P}(B_2\setminus B_1)\rvert \\
	&\leq \mathbb{P}(B_1 \,\Delta\, B_2),
\end{align*}
the last step holding by the triangle inequality. The second inequality holds because $\hat{E}_j(\lambda) \,\Delta\, E_j(\lambda)$ is the union of the events
\begin{align*}
\hat{E}_j(\lambda)\setminus E_j(\lambda) &= \{\Phi_j < \lambda \leq \hat{\Phi}_j\} 
	= \{0 < \lambda - \Phi_j \leq \hat{\Phi}_j - \Phi_j\}\\
E_j(\lambda)\setminus \hat{E}_j(\lambda) &= \{\hat{\Phi}_j < \lambda \leq \Phi_j\}
	= \{\hat{\Phi}_j - \Phi_j < \lambda - \Phi_j \leq 0\}.
\end{align*}
Together, these imply $\hat{E}_j(\lambda) \,\Delta\, E_j(\lambda) \subseteq \{\lvert\Phi_j - \lambda\rvert \leq \lvert\hat{\Phi}_j - \Phi_j\rvert\}$, from which the second inequality follows. The third inequality is a direct consequence of the Lipschitz condition, \cref{ass:lipschitz}. The subsequent equality is a direct application of the law of total expectation, conditioning on the estimated CATE function $\hat{\tau}_{A^c}$. The critical observation is that $\hat{\tau}_{A^c}$ is trained on samples in $A^c$, while $\Phi_j$ depends only on samples in $A$. In particular, letting $C = L\lVert\hat{\tau}_{A^c} - \tau\rVert$, we have
\begin{align*}
\mathbb{P}\left(\lvert\Phi_j - \lambda\rvert \leq L\lVert\hat{\tau}_{A^c} - \tau\rVert\right) &= \int_{\lambda - C}^{\lambda + C} f(\varphi \mid \hat{\tau}_{A^c})\, d\varphi
	\leq M\left(\lambda + C - [\lambda - C]\right)
	= 2MC,
\end{align*}
where the conditional density of $\Phi_j$ given $\hat{\tau}_{A^c}$, $f(\varphi\mid\hat{\tau}_{A^c})$, is uniformly bounded above by $M$ by assumption. Applying the conditional expectation $\mathbb{E}[\cdot\mid\hat{\tau}_{A^c}]$ establishes the result.
\end{proof}

Let $\hat{\pi}^*_{j,m}(\lambda) = \frac{1}{2B}\sum_{b=1}^{2B} \1\!\left(j \in \hat{S}_\lambda(\bX_{A_b}, \tau(\bX_{A_b}))\right)$ denote the estimated oracle selection probability, i.e., the version of $\hat{\pi}_{j,m}(\lambda)$ that would result if $\tau$ were directly observed.

\begin{lemma}[Variance bound]\label{lem:variance}
$\var(\hat{\pi}_{j,m}^*(\lambda)) \leq \frac{1}{8B} + \frac{m^2}{n}$ for any $B$, $j$, $\lambda$, and $m\leq \lfloor n/2\rfloor$.
\end{lemma}

\begin{proof}
Fix $B$, $j$, $\lambda$, and $m$, and set $\hat{\pi}^*=\hat{\pi}^*_{j,m}(\lambda)$. For $b\in[2B]$, define $E_b^* = \{j\in\hat{S}_\lambda(\bX_{A_b}, \tau(\bX_{A_b}))\}$, where each $A_b$ is drawn as in \cref{alg:selection_probabilities}. Then
\begin{equation}\label{eq:var_bound}
\begin{aligned}
\var(\hat{\pi}^*) &= \var\left(\frac{1}{2B}\sum_{b=1}^{2B} \1(E_b^*)\right) \\
&= \frac{1}{4B^2}\sum_{b=1}^{2B} \var(\1(E_b^*)) + \frac{1}{4B^2}\sum_{b\neq b'} \cov(\1(E_b^*),\1(E_{b'}^*)) \\
&\leq \frac{1}{8B} + \frac{1}{4B^2}\sum_{b\neq b'} \cov(\1(E_b^*),\1(E_{b'}^*)).
\end{aligned}
\end{equation}
The inequality holds because $\var(\1(E))\leq 1/4$ for any event $E$. For the covariance terms, fix $b\neq b'$ and set $C = \{A_b \cap A_{b'} = \emptyset\}$.
By conditioning on $C$ and $C^c$,
\begin{align*}
\cov(\1(E_b^*),\1(E_{b'}^*))
&= \mathbb{P}(E_b^* \cap E_{b'}^*) - \mathbb{P}(E_b^*)\mathbb{P}(E_{b'}^*) \\
&= \mathbb{P}(E_b^* \cap E_{b'}^*\mid C^c)\mathbb{P}(C^c)
    + \mathbb{P}(E_b^* \cap E_{b'}^*\mid C)\mathbb{P}(C)
    - \mathbb{P}(E_b^*)\mathbb{P}(E_{b'}^*) \\
&\leq \mathbb{P}(C^c)
    + \mathbb{P}(E_b^* \cap E_{b'}^*\mid C)
    - \mathbb{P}(E_b^*)\mathbb{P}(E_{b'}^*) \\
&= \mathbb{P}(C^c)
    + \mathbb{P}(E_b^*\mid C)\mathbb{P}(E_{b'}^*\mid C)
    - \mathbb{P}(E_b^*)\mathbb{P}(E_{b'}^*) \\
&= \mathbb{P}(C^c).
\end{align*}
The third equality holds because, under the assumption of iid samples, $E_b^*$ and $E_{b'}^*$ are conditionally independent given $A_b$ and $A_{b'}$ with $A_b\cap A_{b'}=\emptyset$. The iid assumption further implies that the corresponding conditional probabilities do not depend on the realized subsets, so the factorization remains valid when conditioning only on $C$. For the last equality, the iid assumption implies that
$\mathbb{P}(E_b^*\mid A_b=S)$ is the same for all $S\subseteq[n]$ of size $m$.
Moreover, the conditional distribution of $A_b$ given $C$ is uniform over
subsets of $[n]$ of size $m$. Hence $\mathbb{P}(E_b^*\mid C)=\mathbb{P}(E_b^*)$, and similarly for $b'$.

To bound $\mathbb{P}(C^c)$, observe that $\mathbb{P}(i\in A_{b'}\mid A_b)\leq m/n$
for every $i\in A_b$. Specifically, if $b$ and $b'$ belong to the same complementary pair in \cref{alg:selection_probabilities}, then $A_b\cap A_{b'}=\emptyset$
and the probability is $0$; otherwise $A_{b'}$ is drawn independently of $A_b$, so
$\mathbb{P}(i\in A_{b'}\mid A_b)=\mathbb{P}(i\in A_{b'})=m/n$. Applying the tower
property and a union bound,
\begin{align*}
\mathbb{P}(C^c)
&= \mathbb{E}\!\left[\mathbb{P}(A_b \cap A_{b'} \neq \emptyset \mid A_b)\right] 
\leq \mathbb{E}\!\bigg[\sum_{i \in A_b} \mathbb{P}(i \in A_{b'} \mid A_b)\bigg]
\leq \mathbb{E}\!\bigg[\sum_{i \in A_b} \frac{m}{n}\bigg]
= \frac{m^2}{n}.
\end{align*}
Thus, $\cov(\1(E_b^*),\1(E_{b'}^*))\leq m^2/n$ for all $b\neq b'$. Plugging this
into \cref{eq:var_bound} gives
\begin{equation*}
\var(\hat{\pi}^*)
\leq \frac{1}{8B} + \frac{2B(2B-1)m^2}{4B^2n}
\leq \frac{1}{8B} + \frac{m^2}{n}. \qedhere
\end{equation*}
\end{proof}

\begin{proof}[Proof of \cref{thrm:convergence}]
Let $L_f$ be the Lipschitz constant for $f$.
\begin{align*}
\mathbb{E}\lvert \hat{\mathcal{E}}^\gamma\, \Delta\, \mathcal{E}^\gamma_* \rvert &\leq \mathbb{E}\left\lvert \left\{ j : \left\lvert \textstyle\int f(\pi_j) \, d\mu - \gamma \right\rvert
    \leq \left\lvert \textstyle\int f(\hat{\pi}_j) - f(\pi_j) \, d\mu \right\rvert \right\} \right\rvert \\
    &\leq \sum_{j=1}^p \mathbb{P}\left(\varepsilon \leq \left\lvert\textstyle\int f(\hat{\pi}_j) - f(\pi_j)\, d\mu\right\rvert\right) \\
    &\leq \sum_{j=1}^p\frac{1}{\varepsilon}\mathbb{E}\left\lvert\textstyle\int f(\hat{\pi}_j) - f(\pi_j)\, d\mu\right\rvert \\
    &\leq \frac{L_f}{\varepsilon}\sum_{j=1}^p\int\mathbb{E}\lvert \hat{\pi}_j - \pi_j\rvert\, d\mu \\
    &\leq \frac{L_f}{\varepsilon}\sum_{j=1}^p \int\mathbb{E}\lvert\hat{\pi}_j - \mathbb{E}(\hat{\pi}_j)\rvert + \lvert\mathbb{E}(\hat{\pi}_j) - \pi_j\rvert\, d\mu \\
    &\leq \frac{L_f}{\varepsilon}\sum_{j=1}^p \int \sqrt{\var(\hat{\pi}_j)} + \lvert\mathbb{E}(\hat{\pi}_j) - \pi_j\rvert\, d\mu
\end{align*}
By \cref{thrm:bias}, we have $\lvert\mathbb{E}(\hat{\pi}_j) - \pi_j\rvert \leq 2ML\mathbb{E}\lVert\hat{\tau}_{A^c} - \tau\rVert$ for all $j$. For the variance term, fix $j$ and $\lambda$ and set $\hat{\pi}=\hat{\pi}_j(\lambda)$. Define $\hat{E}_b = \{j\in\hat{S}_\lambda(\bX_{A_b},\hat{\tau}_{A_b^c})\}$ and $E_b = \{j\in\hat{S}_\lambda(\bX_{A_b},\tau)\}$. Then
\begin{align*}
\var(\hat{\pi}) &= \var(\hat{\pi}^* + \hat{\pi} - \hat{\pi}^*)
    \leq 2\var(\hat{\pi}^*) + 2\var(\hat{\pi} - \hat{\pi}^*)
\end{align*}
An upper bound for the first term, $\var(\hat{\pi}^*)$, is provided by \cref{lem:variance}. For the second term,
\begin{align*}
\var(\hat{\pi} - \hat{\pi}^*) &\leq \mathbb{E}\lvert\hat{\pi} - \hat{\pi}^*\rvert \\
    &= \mathbb{E}\bigg\lvert\frac{1}{2B}\sum_{b=1}^{2B} \1(\hat{E}_b) - \1(E_b)\bigg\rvert \\
    &\leq \frac{1}{2B}\sum_{b=1}^{2B} \mathbb{E}\left\lvert\1(\hat{E}_b) - \1(E_b)\right\rvert \\
    &= \mathbb{P}(\hat{E}_b\Delta E_b) \\
    &\leq 2ML\mathbb{E}\lVert\hat{\tau}_{A^c} - \tau\rVert.
\end{align*}
The first inequality holds because $\lvert\hat{\pi} - \hat{\pi}^*\rvert \leq 1$ and hence $\var(\hat{\pi} - \hat{\pi}^*) \leq \mathbb{E}[(\hat{\pi} - \hat{\pi}^*)^2] \leq \mathbb{E}\lvert\hat{\pi} - \hat{\pi}^*\rvert$. The second inequality is the triangle inequality. The second equality holds because $\lvert\1(\hat{E}_b) - \1(E_b)\rvert = \1(\hat{E}_b\Delta E_b)$, and the third inequality was established in the proof of \cref{thrm:bias}. Therefore,
\begin{align*}
\sqrt{\var(\hat{\pi})} &\leq \sqrt{\frac{1}{4B} + \frac{2m^2}{n} + 4ML\mathbb{E}\lVert\hat{\tau}_{A^c} - \tau\rVert}
\end{align*}
and hence,
\begin{align*}
\mathbb{E}\lvert \hat{\mathcal{E}}^\gamma\, \Delta\, \mathcal{E}^\gamma_* \rvert &\leq \frac{L_f p}{\varepsilon}\int\left(\sqrt{\frac{1}{4B} + \frac{2m^2}{n} + 4ML\mathbb{E}\lVert\hat{\tau}_{A^c} - \tau\rVert} + 2ML\mathbb{E}\lVert\hat{\tau}_{A^c} - \tau\rVert\right)\, d\mu \\
    &= \frac{L_f p}{\varepsilon}\left(\sqrt{\frac{1}{4B} + \frac{2m^2}{n} + 4ML\mathbb{E}\lVert\hat{\tau}_{A^c} - \tau\rVert} + 2ML\mathbb{E}\lVert\hat{\tau}_{A^c} - \tau\rVert\right).
\end{align*}
The equality holds since $\mu$ is a probability measure on $\Lambda$ and the integrand is independent of $\lambda$.
\end{proof}

\subsection{\cref{sec:theory_fdr} proofs}\label{sup_sec:false_discovery}

Throughout this section, we fix a subsample size $m\leq \lfloor n/2\rfloor$ and suppress $m$ from notation. For example, we write $\hat{\pi}_j(\lambda)$ instead of $\hat{\pi}_{j,m}(\lambda)$. The following quantities are used in the proofs of
\cref{thrm:fdr,thrm:oracle_fdr}. Fix $B \geq 2$. For $\lambda \in \Lambda$ and
$j \in [p]$, define the \emph{simultaneous selection probabilities}
\begin{align*}
    \hat{\xi}_j(\lambda)
    &= \frac{1}{B}\sum_{b=1}^B
        \1\!\left(j \in \hat{S}_\lambda(\bX_{A_{2b-1}}, \hat{\tau}_{A_{2b-1}^c})\right)
        \1\!\left(j \in \hat{S}_\lambda(\bX_{A_{2b}}, \hat{\tau}_{A_{2b}^c})\right), \\
    \xi_j(\lambda)
    &= \frac{1}{B}\sum_{b=1}^B
        \1\!\left(j \in \hat{S}_\lambda(\bX_{A_{2b-1}}, \tau)\right)
        \1\!\left(j \in \hat{S}_\lambda(\bX_{A_{2b}}, \tau)\right).
\end{align*}
Since each summand lies in $\{0,1\}$, a direct expansion of $\hat{\xi}_j(\lambda)$ gives
\begin{align*}
    0 \leq
    \frac{1}{B}\sum_{b=1}^B
    \Bigl(1 - \1\bigl(j \in \hat{S}_\lambda(\bX_{A_{2b-1}}, \hat{\tau}_{A_{2b-1}^c})\bigr)\Bigr)
    \Bigl(1 - \1\bigl(j \in \hat{S}_\lambda(\bX_{A_{2b}}, \hat{\tau}_{A_{2b}^c})\bigr)\Bigr)
    = 1 - 2\hat{\pi}_j(\lambda) + \hat{\xi}_j(\lambda),
\end{align*}
and hence $\hat{\pi}_j(\lambda) \leq \tfrac{1}{2}(1 + \hat{\xi}_j(\lambda))$.
Replacing $\hat{\tau}$ with $\tau$ yields the oracle analogue
$\hat{\pi}^*_j(\lambda) \leq \tfrac{1}{2}(1 + \xi_j(\lambda))$.
Since $f$ is nondecreasing and satisfies $f\bigl(\tfrac{1}{2}(1+x)\bigr) = x^3$ for $x \geq 0$,
\begin{equation}\label{eq:f_xi_bound}
    f\!\left(\hat{\pi}_j(\lambda)\right)
    \leq
    f\!\left(\tfrac{1}{2}(1 + \hat{\xi}_j(\lambda))\right)
    = \hat{\xi}_j(\lambda)^3,
\end{equation}
and likewise $f\bigl(\hat{\pi}^*_j(\lambda)\bigr) \leq \xi_j(\lambda)^3$.

For readability, define
\begin{align*}
    \Delta_r = \Bigl\{\bm{k} \in \mathbb{Z}^B : k_b \geq 0,\;
    {\textstyle\sum_{b=1}^B} k_b = r\Bigr\},
\end{align*}
and $ N_k = \sum_{b=1}^B \1(k_b \neq 0)$, and denote multinomial coefficients by $\tbinom{r}{k_1,\ldots,k_B} = r!/(k_1!\cdots k_B!)$.

The following lemma is the \css{} analogue of Lemma~S2.1 in
\citet{melikechi2026integrated}. The proof is the same multinomial expansion
argument, the only change being that it invokes the integrated form of
\cref{ass:exchangeability} rather than the pointwise form used in that work.

\begin{lemma}\label{lem:moment_bound}
Fix $r \in \mathbb{N}$. If \cref{ass:exchangeability}a holds for all $r' \in [r]$, then
\begin{align}\label{eq:moment_bound_hat}
    \max_{j \in \mathcal{E}^c}
    \int_\Lambda \mathbb{E}\!\left[\hat{\xi}_j(\lambda)^r\right]\mu(d\lambda)
    \leq
    \frac{1}{B^r}
    \sum_{\bm{k} \in \Delta_r}
    \binom{r}{k_1,\ldots,k_B}
    \int_\Lambda
    \!\left(\frac{\hat{q}(\lambda)}{p}\right)^{\!2N_k}
    \!\mu(d\lambda).
\end{align}
If \cref{ass:exchangeability}b holds for all $r' \in [r]$, then \eqref{eq:moment_bound_hat} holds with $\xi_j$ and $q(\lambda)$ in place of $\hat{\xi}_j$ and $\hat{q}(\lambda)$.
\end{lemma}

\begin{proof}
Fix $j \in \mathcal{E}^c$ and define $0^0 = 1$. Let
\begin{align*}
    U_b(\lambda)
    = \1\!\left(j \in \hat{S}_\lambda(\bX_{A_{2b-1}}, \hat{\tau}_{A_{2b-1}^c})\right)
    \1\!\left(j \in \hat{S}_\lambda(\bX_{A_{2b}}, \hat{\tau}_{A_{2b}^c})\right)
\end{align*}
so that $\hat{\xi}_j(\lambda) = B^{-1}\sum_{b=1}^B U_b(\lambda)$. By the multinomial theorem,
\begin{align*}
    \hat{\xi}_j(\lambda)^r
    = \frac{1}{B^r}\sum_{\bm{k} \in \Delta_r}
    \binom{r}{k_1,\ldots,k_B}
    \prod_{b=1}^B U_b(\lambda)^{k_b}
    = \frac{1}{B^r}\sum_{\bm{k} \in \Delta_r}
    \binom{r}{k_1,\ldots,k_B}
    \prod_{b=1}^B U_b(\lambda)^{\1(k_b \neq 0)}.
\end{align*}
The last equality holds because $U_b(\lambda) \in \{0,1\}$, and hence $U_b(\lambda)^{k_b} = U_b(\lambda)$ whenever $k_b > 0$. Next, since \cref{ass:exchangeability}a holds for all $r' \leq r$ and since every $\bm{k} \in \Delta_r$ has at most $r$ nonzero entries,
\begin{align*}
    \int_\Lambda \mathbb{E}\bigg[\prod_{b=1}^B U_b(\lambda)^{\1(k_b \neq 0)}\bigg]\mu(d\lambda)
    &= \int_\Lambda \mathbb{P}\!\bigg(j \in \bigcap_{b : k_b \neq 0}
        \Bigl[\hat{S}_\lambda(\bX_{A_{2b-1}}, \hat{\tau}_{A_{2b-1}^c}) \cap \hat{S}_\lambda(\bX_{A_{2b}}, \hat{\tau}_{A_{2b}^c})\Bigr]\!\bigg)\mu(d\lambda) \\
    &\leq \int_\Lambda \!\left(\frac{\hat{q}(\lambda)}{p}\right)^{\!2N_k}\!\mu(d\lambda).
\end{align*}
Therefore,
\begin{align*}
    \int_\Lambda \mathbb{E}\!\left[\hat{\xi}_j(\lambda)^r\right]\mu(d\lambda)
    &= \frac{1}{B^r}\sum_{\bm{k} \in \Delta_r}
    \binom{r}{k_1,\ldots,k_B}
    \int_\Lambda \mathbb{E}\bigg[\prod_{b=1}^B U_b(\lambda)^{\1(k_b \neq 0)}\bigg]\mu(d\lambda) \\
    &\leq \frac{1}{B^r}\sum_{\bm{k} \in \Delta_r}
    \binom{r}{k_1,\ldots,k_B}
    \int_\Lambda \!\left(\frac{\hat{q}(\lambda)}{p}\right)^{\!2N_k}\!\mu(d\lambda).
\end{align*}
Taking the maximum over $j \in \mathcal{E}^c$ gives \eqref{eq:moment_bound_hat}.
\end{proof}

\begin{proof}[Proof of \cref{thrm:fdr}]
Fix $j \in \mathcal{E}^c$ and $\gamma>0$. By Markov's inequality,
\begin{align*}
    \mathbb{P}(j \in \hat{\mathcal{E}}_\gamma)
    = \mathbb{P}\!\left(\int_\Lambda f(\hat{\pi}_j(\lambda))\,\mu(d\lambda) \geq \gamma\right)
    \leq \frac{1}{\gamma}\int_\Lambda \mathbb{E}\!\left[f(\hat{\pi}_j(\lambda))\right]\mu(d\lambda)
    \leq \frac{1}{\gamma}\int_\Lambda \mathbb{E}\!\left[\hat{\xi}_j(\lambda)^3\right]\mu(d\lambda),
\end{align*}
where the last inequality holds by \eqref{eq:f_xi_bound}. Applying \cref{lem:moment_bound} with $r = 3$ under \cref{ass:exchangeability}a,
\begin{align*}
    \mathbb{P}(j \in \hat{\mathcal{E}}_\gamma)
    \leq \frac{1}{\gamma B^3}
    \sum_{\bm{k} \in \Delta_3}\binom{3}{k_1,\ldots,k_B}
    \int_\Lambda \!\left(\frac{\hat{q}(\lambda)}{p}\right)^{\!2N_k}\!\mu(d\lambda).
\end{align*}
Therefore,
\begin{align*}
    \mathbb{E}|\hat{\mathcal{E}}_\gamma \cap \mathcal{E}^c|
    = \sum_{j \in \mathcal{E}^c}\mathbb{P}(j \in \hat{\mathcal{E}}_\gamma)
    \leq \frac{p}{\gamma B^3}\sum_{\bm{k} \in \Delta_3}\binom{3}{k_1,\ldots,k_B}\int_\Lambda \!\left(\frac{\hat{q}(\lambda)}{p}\right)^{\!2N_k}\!\mu(d\lambda).
\end{align*}
It remains to evaluate the sum over $\Delta_3=\{\bm{k}\in\mathbb{Z}^B : k_b\geq 0, \sum_{b=1}^B k_b=3\}$. The elements of $\Delta_3$ partition into three groups according to the number of nonzero entries, $N_k$:
\begin{itemize}[itemsep=0.5em, leftmargin=1.5em]
    \item $N_k = 1$: a single entry equal to $3$. There are $B$ such $\bm{k}$, each with $\binom{3}{k_1,\ldots,k_B} = 1$.
    \item $N_k = 2$: one entry equal to 1, another to 2. There are $B(B-1)$ such $\bm{k}$, each with $\binom{3}{k_1,\ldots,k_B} = 3$.
    \item $N_k = 3$: three entries equal to 1. There are $\binom{B}{3}$ such $\bm{k}$, each with $\binom{3}{k_1,\ldots,k_B} = 6$.
\end{itemize}
Substituting,
\begin{align*}
    \mathbb{E}|\hat{\mathcal{E}}_\gamma \cap \mathcal{E}^c|
    &\leq \frac{p}{\gamma B^3}\int_\Lambda
        B\!\left(\frac{\hat{q}(\lambda)}{p}\right)^{\!2}
        + 3B(B-1)\!\left(\frac{\hat{q}(\lambda)}{p}\right)^{\!4}
        + B(B-1)(B-2)\!\left(\frac{\hat{q}(\lambda)}{p}\right)^{\!6}\mu(d\lambda) \\
    &= \frac{1}{\gamma}\int_\Lambda
        \frac{\hat{q}(\lambda)^2}{B^2 p}
        + \frac{3(B-1)\hat{q}(\lambda)^4}{B^2 p^3}
        + \frac{(B-1)(B-2)\hat{q}(\lambda)^6}{B^2 p^5}\,\mu(d\lambda). \qedhere
\end{align*}
\end{proof}

\begin{proof}[Proof of \cref{thrm:oracle_fdr}]
Since $\hat{\mathcal{E}}_\gamma \cap \mathcal{E}^c
\subseteq (\hat{\mathcal{E}}_\gamma \,\Delta\, \mathcal{E}^*_\gamma)
\cup (\mathcal{E}^*_\gamma \cap \mathcal{E}^c)$,
taking expectations gives
\begin{align*}
    \mathbb{E}|\hat{\mathcal{E}}_\gamma \cap \mathcal{E}^c|
    \leq \mathbb{E}|\hat{\mathcal{E}}_\gamma \,\Delta\, \mathcal{E}^*_\gamma|
       + |\mathcal{E}^*_\gamma \cap \mathcal{E}^c|
    \leq \rho + |\mathcal{E}^*_\gamma \cap \mathcal{E}^c|,
\end{align*}
where $\mathcal{E}^*_\gamma$ is deterministic so no expectation is needed on the
second term, and the second inequality is \cref{thrm:convergence}. It remains to
bound $|\mathcal{E}^*_\gamma \cap \mathcal{E}^c|$.

Applying $\1(x \geq \gamma) \leq x/\gamma$ for $x \geq 0$,
\begin{align*}
    |\mathcal{E}^*_\gamma \cap \mathcal{E}^c|
    = \sum_{j \in \mathcal{E}^c}
      \1\!\left(\int_\Lambda f(\pi_j(\lambda))\,\mu(d\lambda) \geq \gamma\right)
    \leq \frac{p}{\gamma}
         \max_{j \in \mathcal{E}^c}
         \int_\Lambda f(\pi_j(\lambda))\,\mu(d\lambda).
\end{align*}
We now bound $f(\pi_j(\lambda))$ using Jensen's inequality. The function
$f(x) = (2x-1)^3\1(x \geq \tfrac{1}{2})$ is convex on $[0,1]$: it is identically
zero on $[0,\tfrac{1}{2}]$, satisfies $f''(x) = 24(2x-1) \geq 0$ on
$[\tfrac{1}{2},1]$, and has $f'(\tfrac{1}{2}) = 0$ from both sides. Since
$\hat{\pi}^*_j(\lambda)$ is an unbiased estimator of $\pi_j(\lambda)$, Jensen's
inequality gives
\begin{align*}
    f(\pi_j(\lambda))
    = f\!\left(\mathbb{E}[\hat{\pi}^*_j(\lambda)]\right)
    \leq \mathbb{E}\!\left[f(\hat{\pi}^*_j(\lambda))\right]
    \leq \mathbb{E}\!\left[\xi_j(\lambda)^3\right],
\end{align*}
where the last step uses the oracle analogue of \eqref{eq:f_xi_bound}.
Integrating over $\Lambda$ and applying \cref{lem:moment_bound} with $r = 3$
under \cref{ass:exchangeability}b yields the same multinomial expansion as in
the proof of \cref{thrm:fdr}, with $q(\lambda)$ in place of $\hat{q}(\lambda)$.
Combining the two bounds gives \eqref{eq:oracle_fdr_bound}.
\end{proof}

\section{Additional Simulation Results}\label{sup_sec:simulations}


In this section, we report additional simulation results.

\begin{figure}[H]
\includegraphics[width=\textwidth]{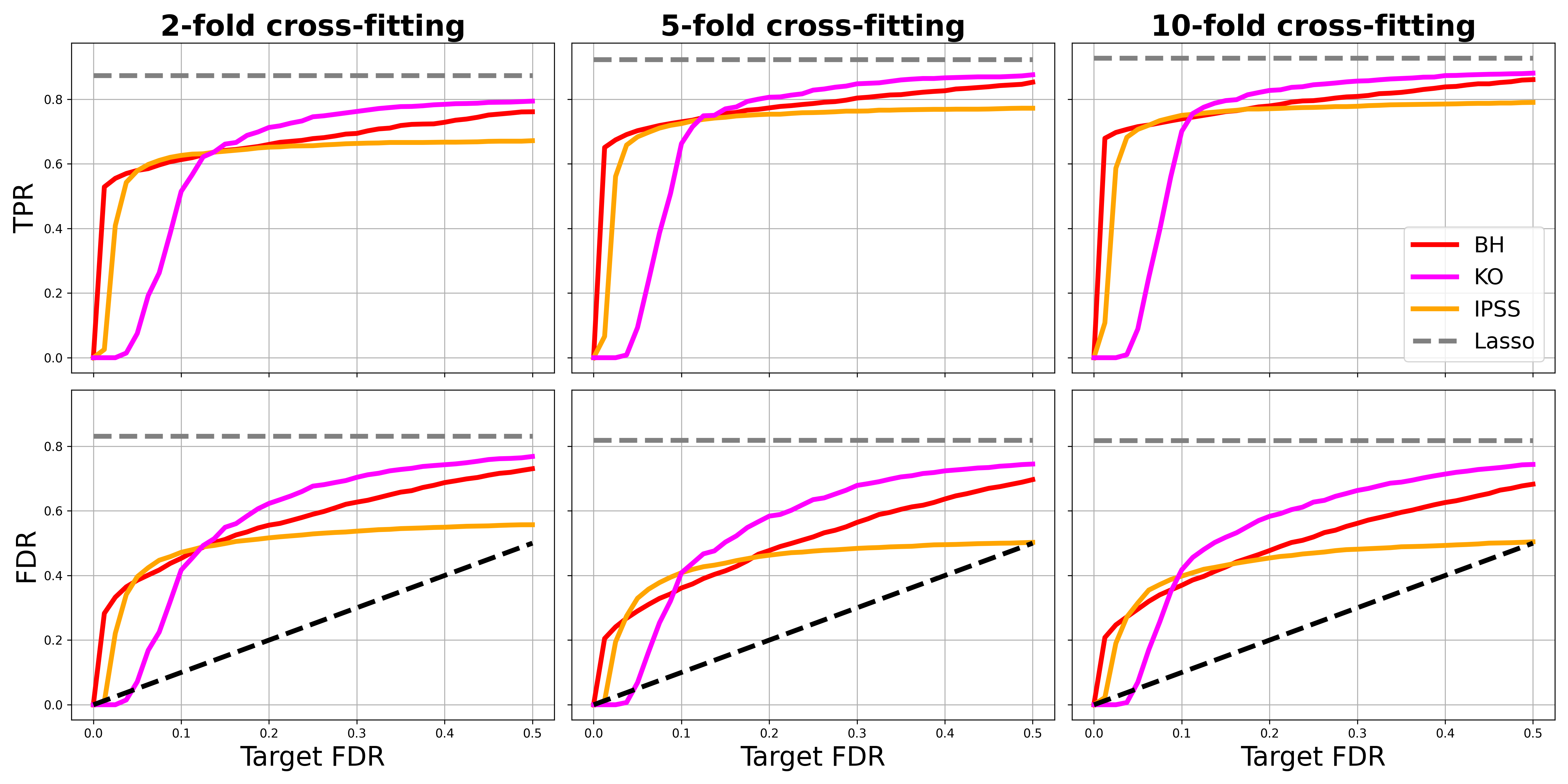}
\caption{\textit{$K$-fold cross-fitting with different $K$}. Results for the $K$-fold cross-fitting approach to effect modifier discovery described in \cref{sec:illustation}. While going from 2-folds to 5 shows mild improvements in performance, there is almost no difference between $K=5$ and 10. In all cases, the FDR is poorly controlled.}
\label{fig:crossfitting}
\end{figure}

\begin{figure}[H]
\includegraphics[width=\textwidth]{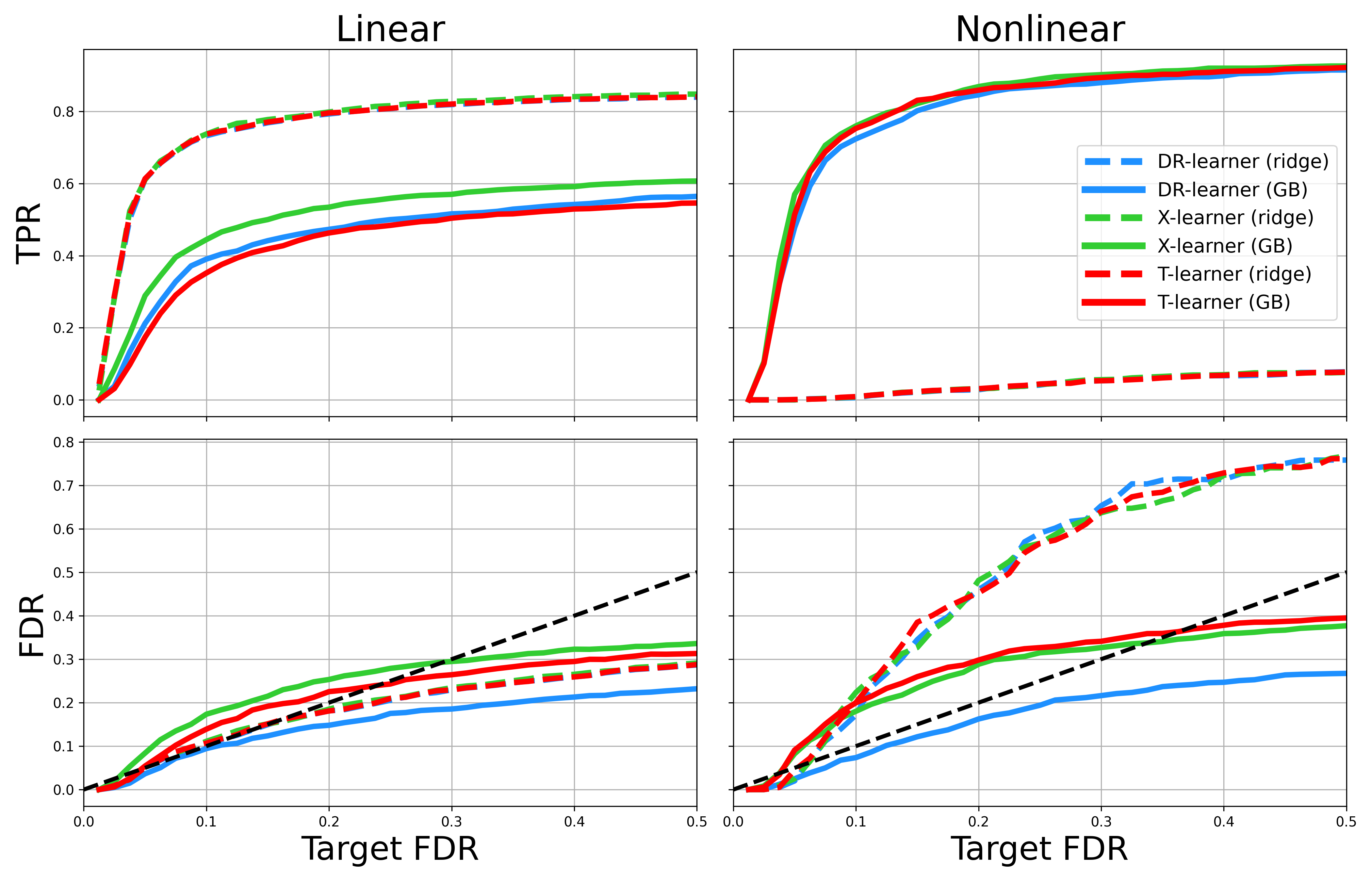}
\caption{\textit{\css{} with different CATE estimators}. \css{} is implemented with six different CATE estimators: two versions each of a DR-learner, a T-learner, and an X-learner, where one version uses ridge regression, and the other uses XGBoost. Data are simulated according to the baseline simulation designs described in \cref{sec:sim_design}; see also the bold parameter values in \cref{tab:sim_params}. For each implementation of \css{}, only the CATE estimator varies; all other inputs are held constant.}
\label{fig:cate_comparison}
\end{figure}

\begin{figure}[H]
\includegraphics[width=\textwidth]{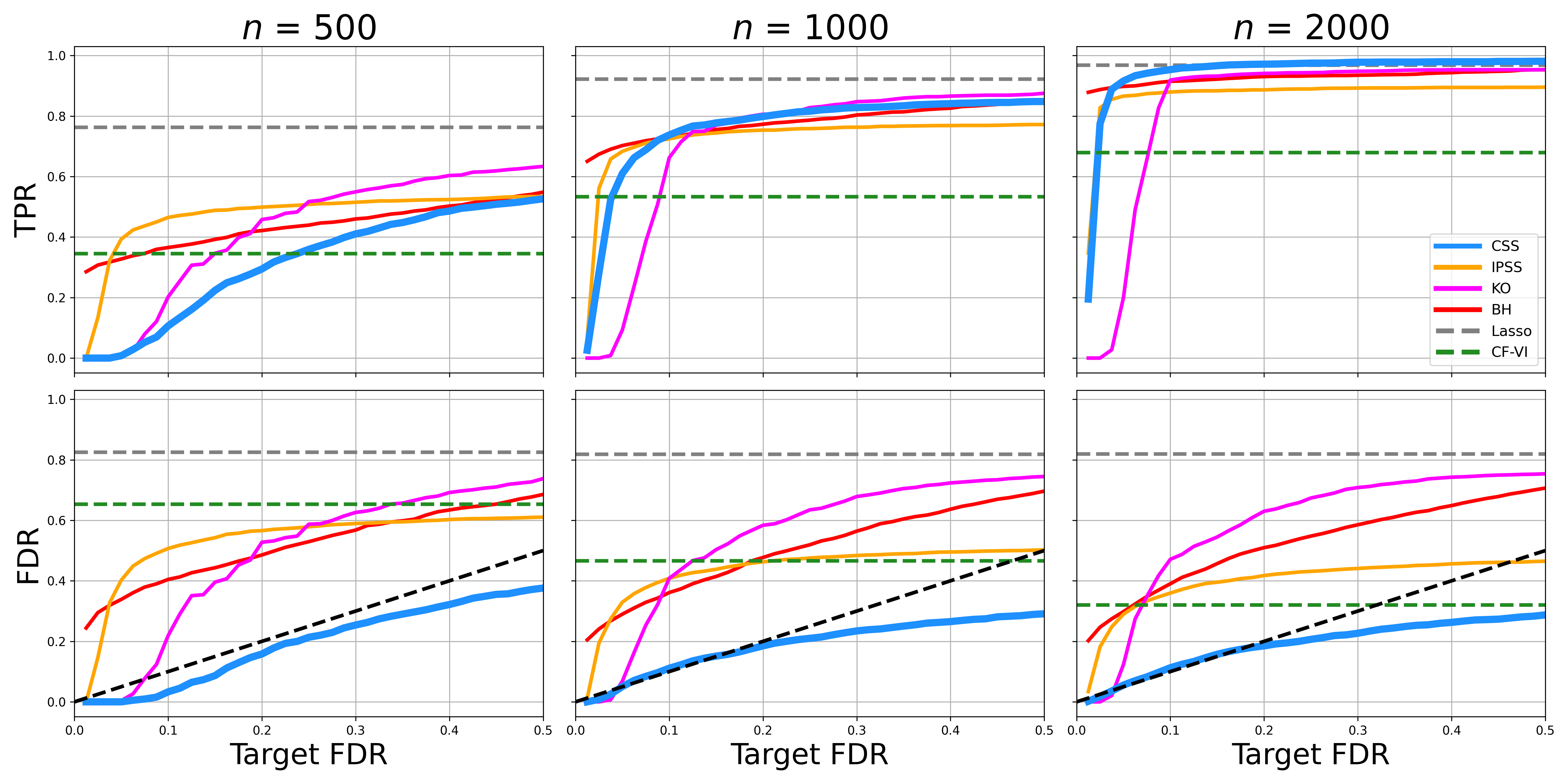}
\caption{\textit{Linear results: sample size, $n$.}}
\label{fig:sup_linear_n}
\end{figure}

\begin{figure}[H]
\includegraphics[width=\textwidth]{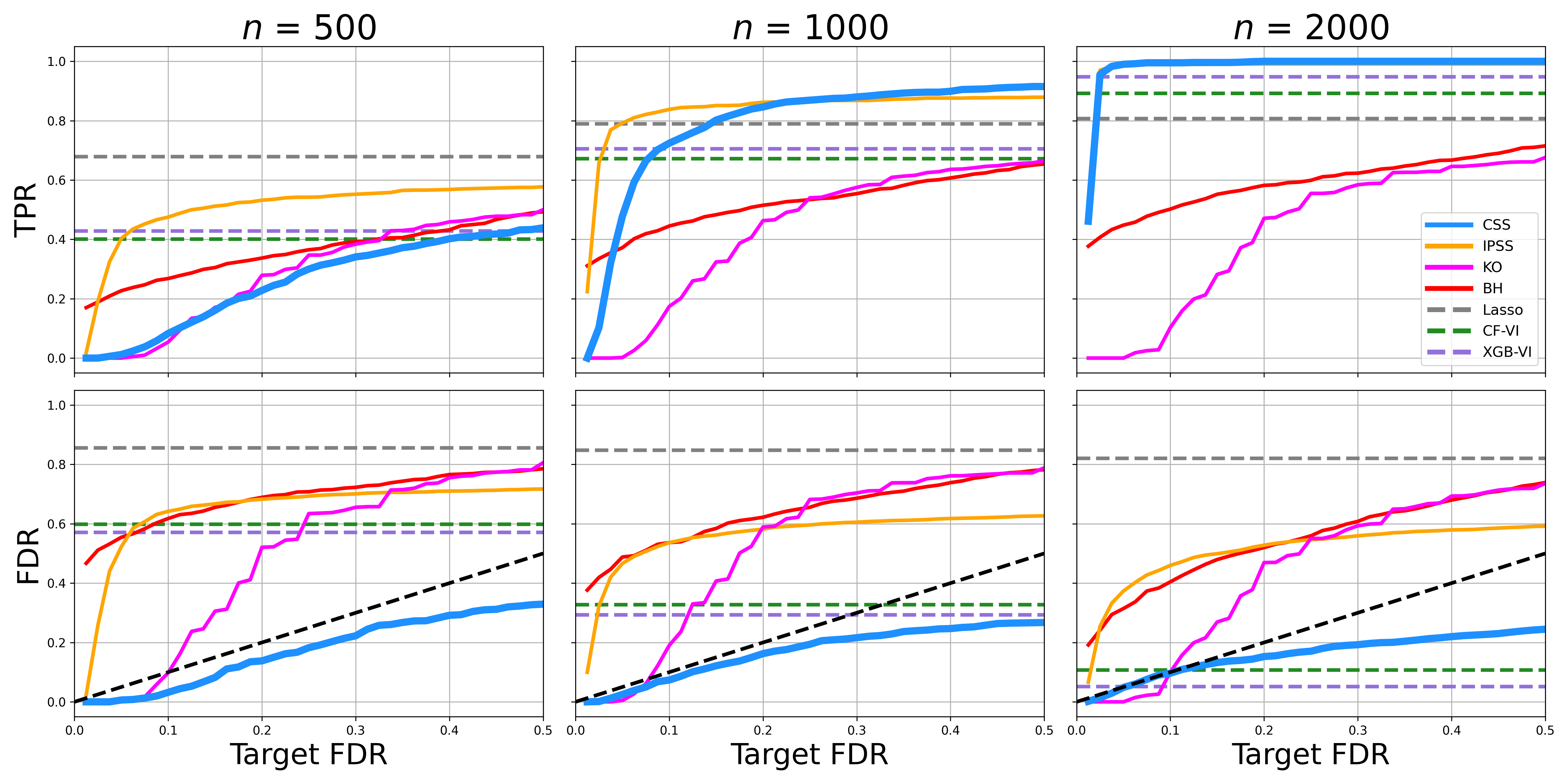}
\caption{\textit{Nonlinear results: sample size, $n$.}}
\label{fig:sup_nonlinear_n}
\end{figure}

\begin{figure}[H]
\includegraphics[width=\textwidth]{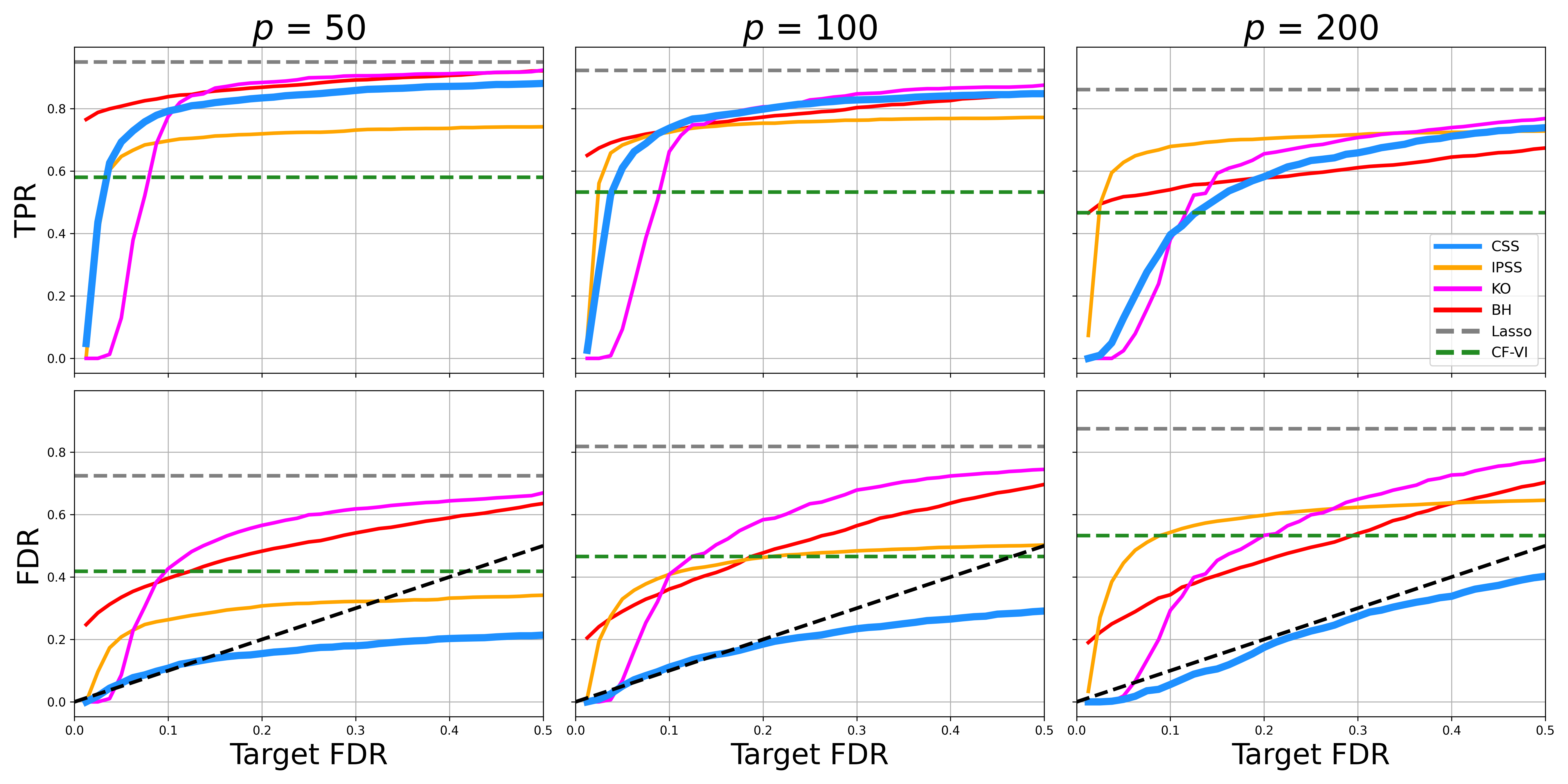}
\caption{\textit{Linear results: number of covariates, $p$.}}
\label{fig:sup_linear_p}
\end{figure}

\begin{figure}[H]
\includegraphics[width=\textwidth]{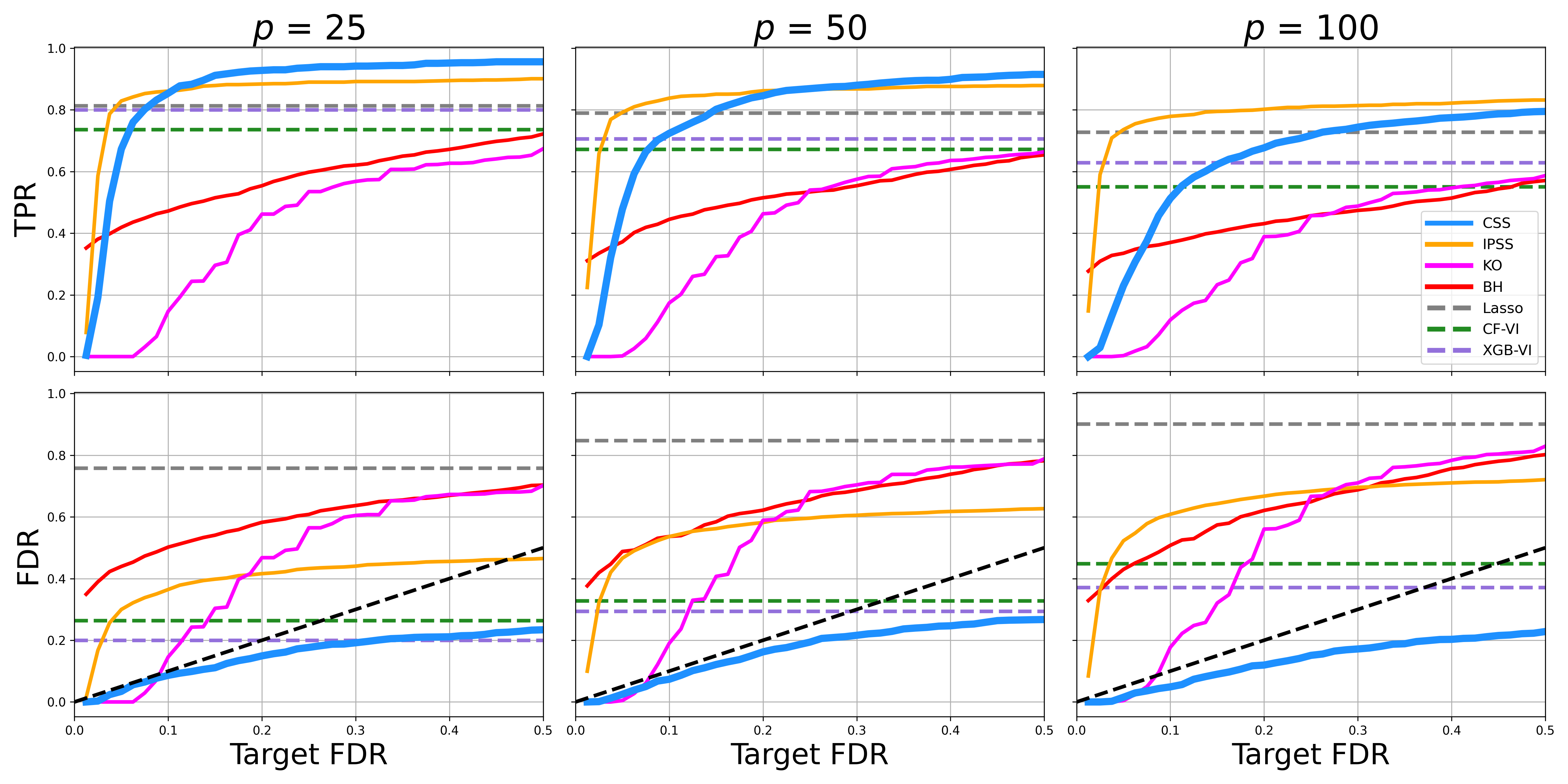}
\caption{\textit{Nonlinear results: number of covariates, $p$.}}
\label{fig:sup_nonlinear_p}
\end{figure}

\begin{figure}[H]
\includegraphics[width=\textwidth]{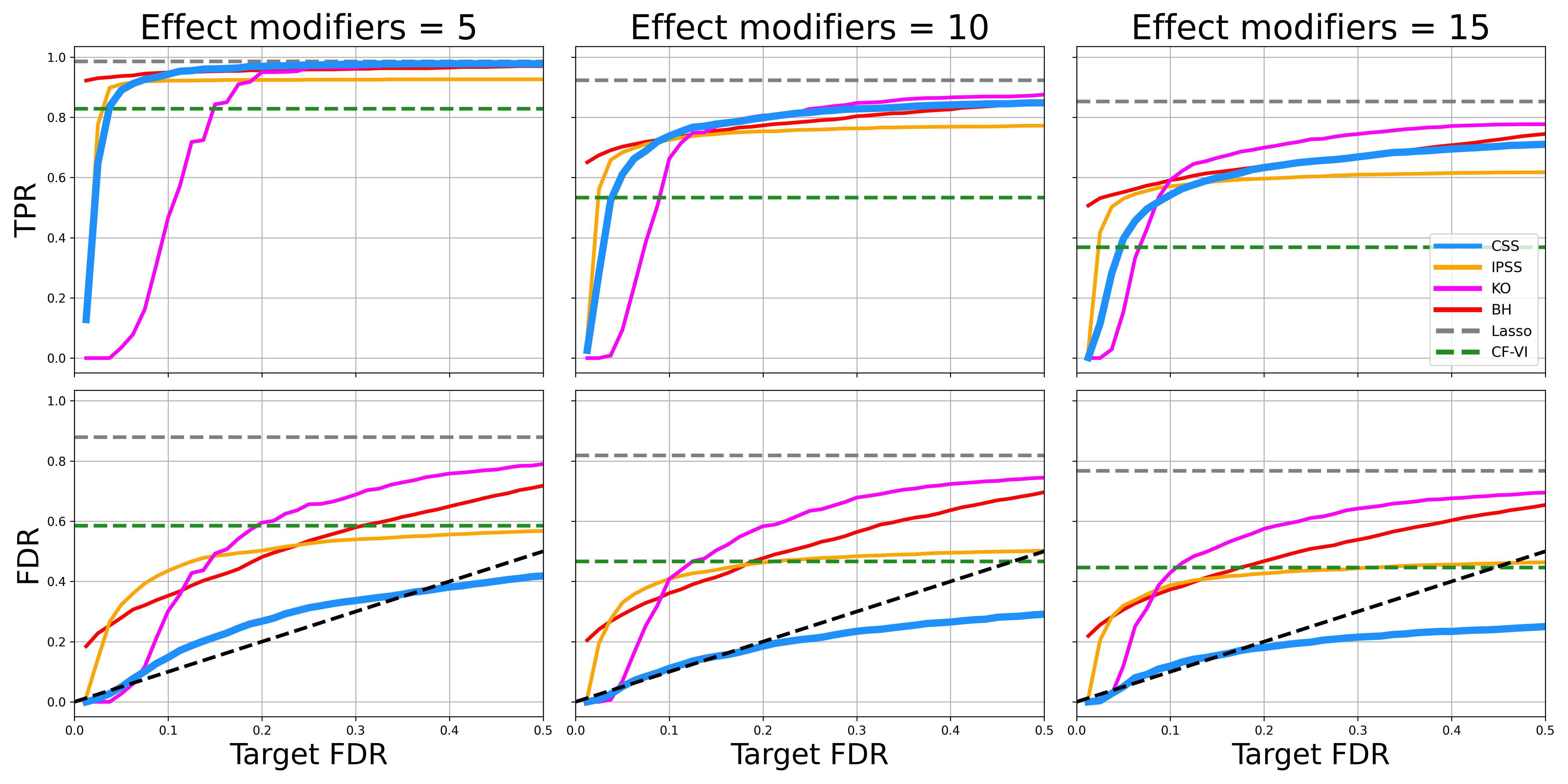}
\caption{\textit{Linear results: number of effect modifiers, $|\mathcal{E}|$.}}
\label{fig:sup_linear_effect_modifiers}
\end{figure}

\begin{figure}[H]
\includegraphics[width=\textwidth]{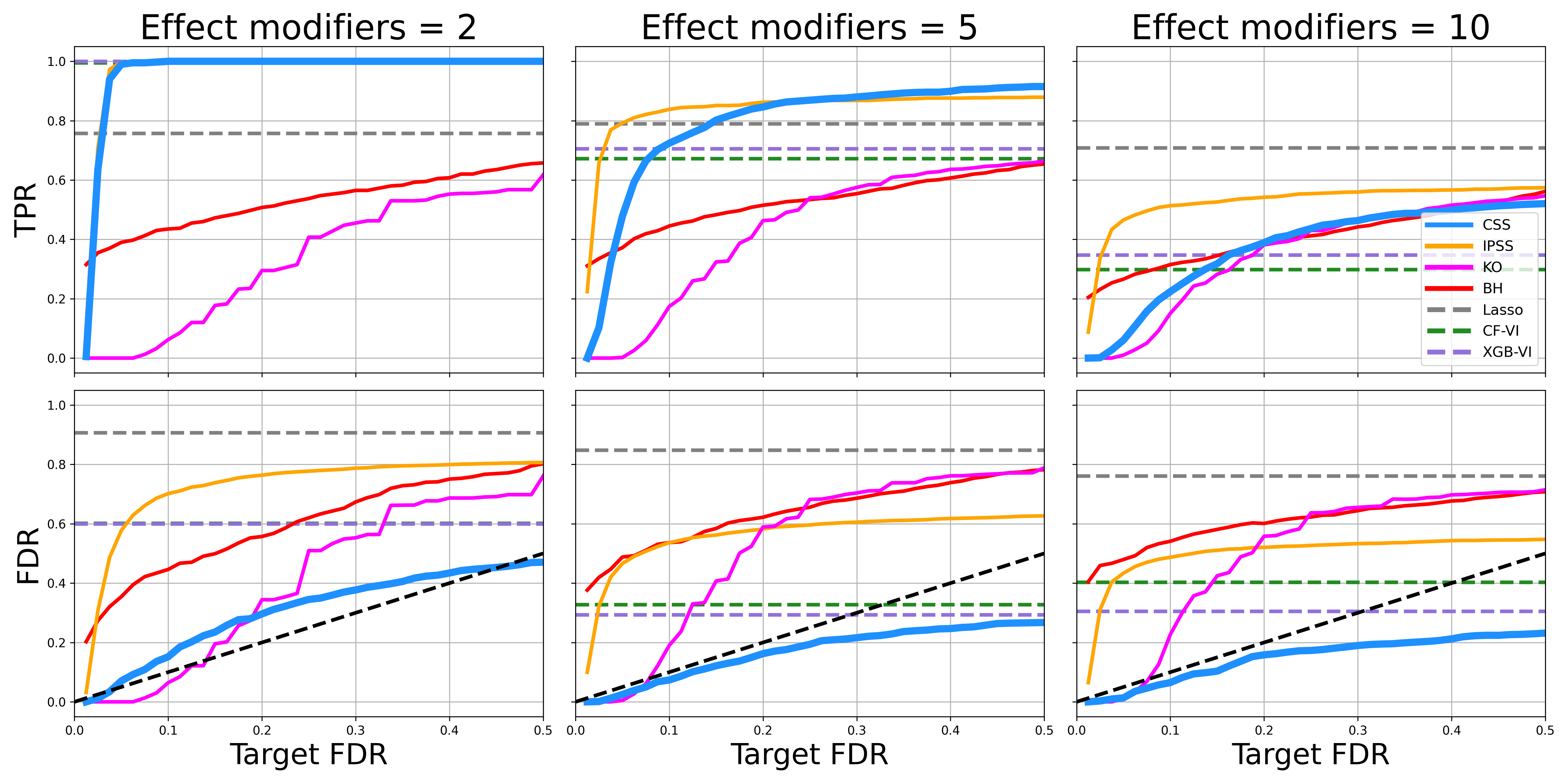}
\caption{\textit{Nonlinear results: number of effect modifiers, $|\mathcal{E}|$.}}
\label{fig:sup_nonlinear_effect_modifiers}
\end{figure}

\begin{figure}[H]
\includegraphics[width=\textwidth]{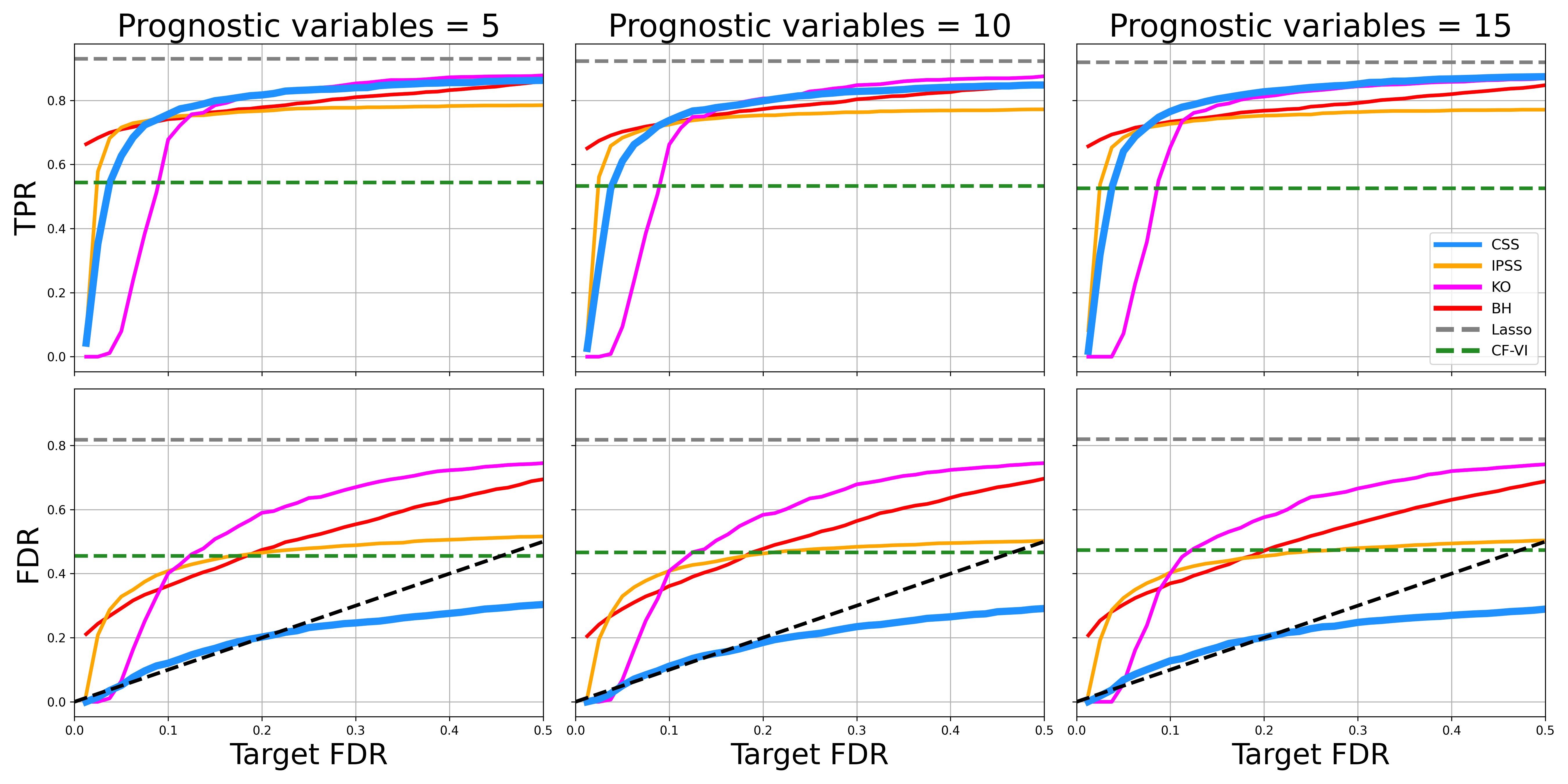}
\caption{\textit{Linear results: number of prognostic variables, $|\mathcal{P}|$.}}
\label{fig:sup_linear_prognostic_vars}
\end{figure}

\begin{figure}[H]
\includegraphics[width=\textwidth]{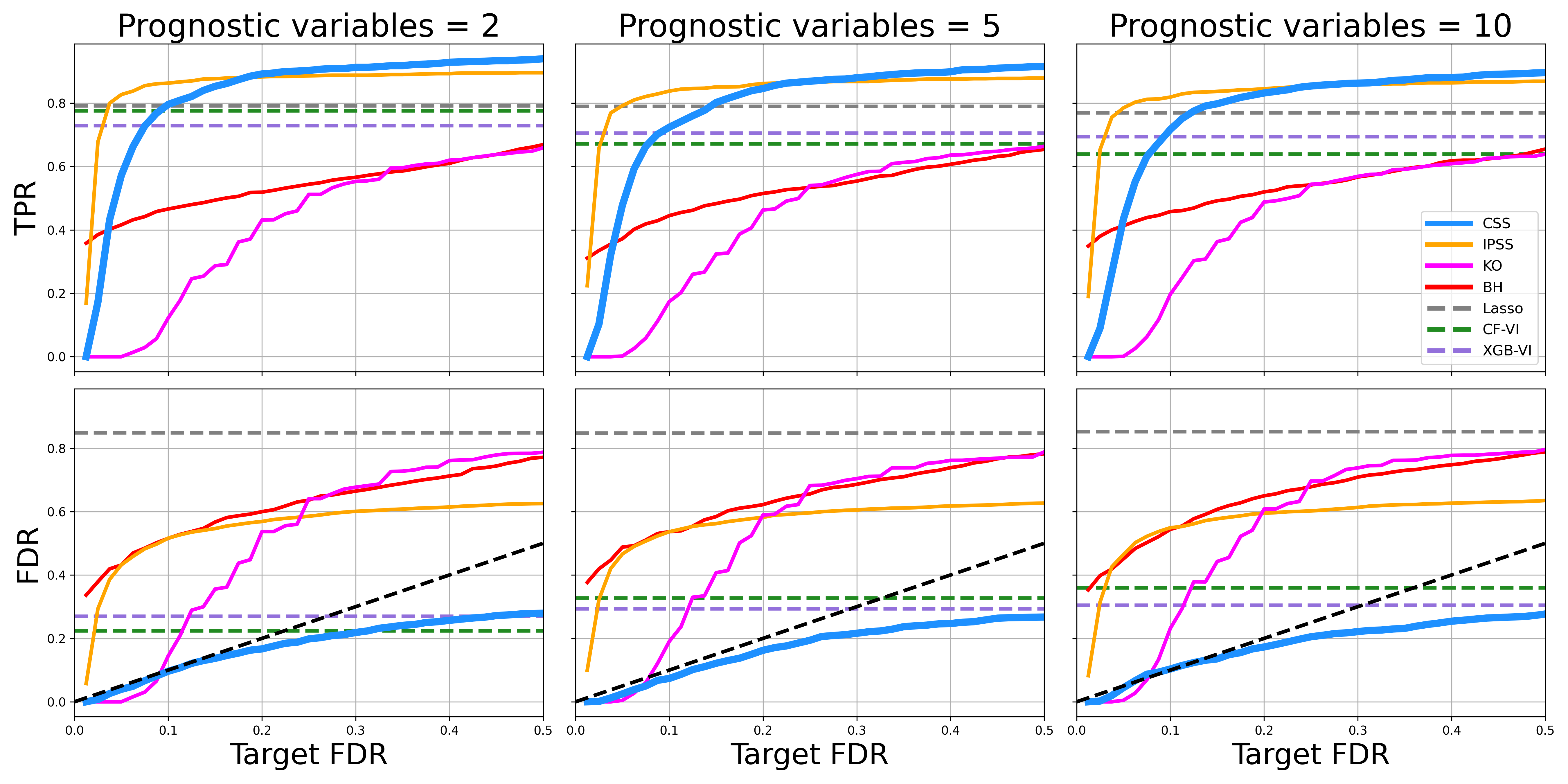}
\caption{\textit{Nonlinear results: number of prognostic variables, $|\mathcal{P}|$.}}
\label{fig:sup_nonlinear_prognostic_vars}
\end{figure}

\begin{figure}[H]
\includegraphics[width=\textwidth]{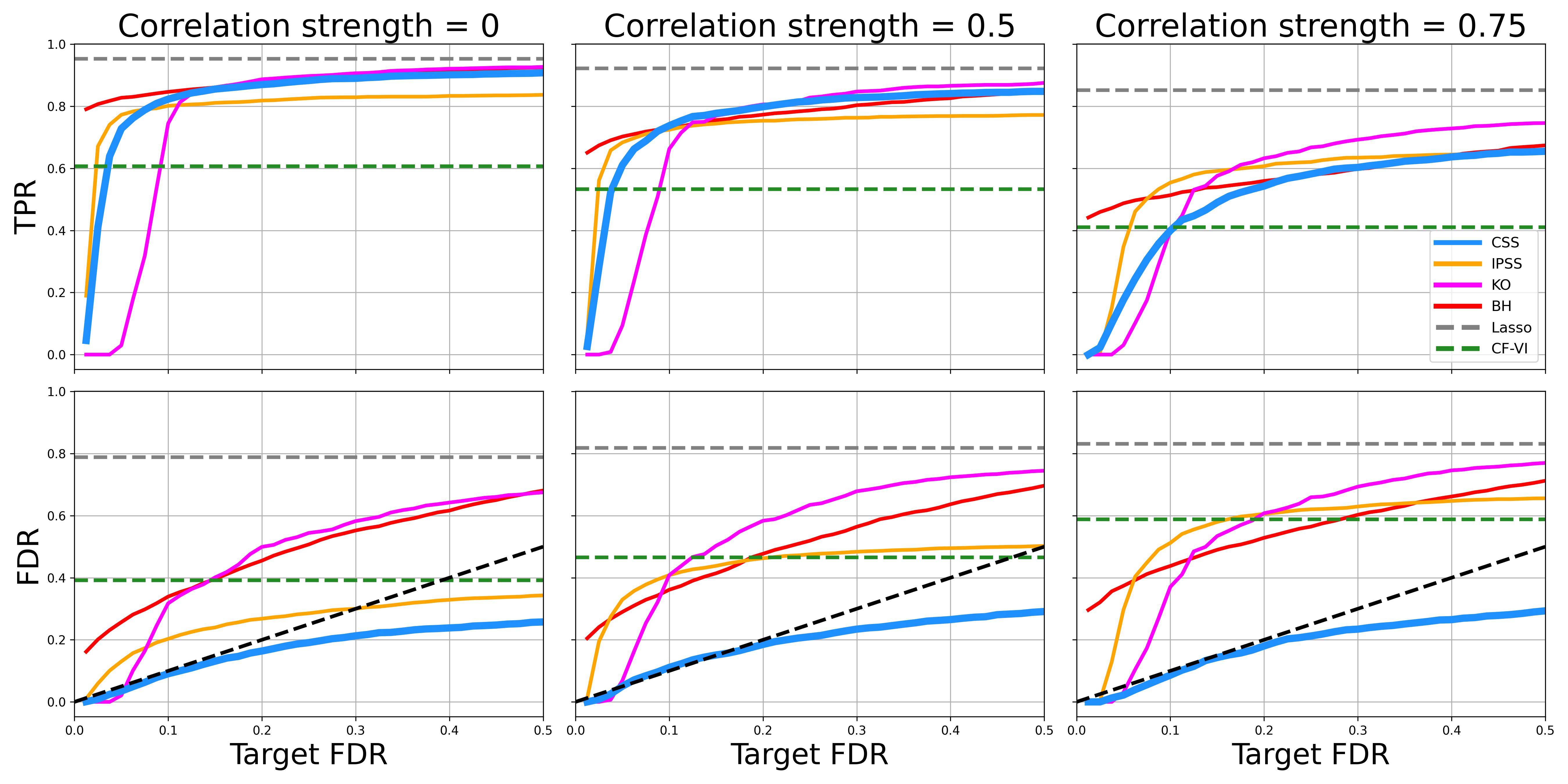}
\caption{\textit{Linear results: covariate correlation, $\rho$.}}
\label{fig:sup_linear_rho}
\end{figure}

\begin{figure}[H]
\includegraphics[width=\textwidth]{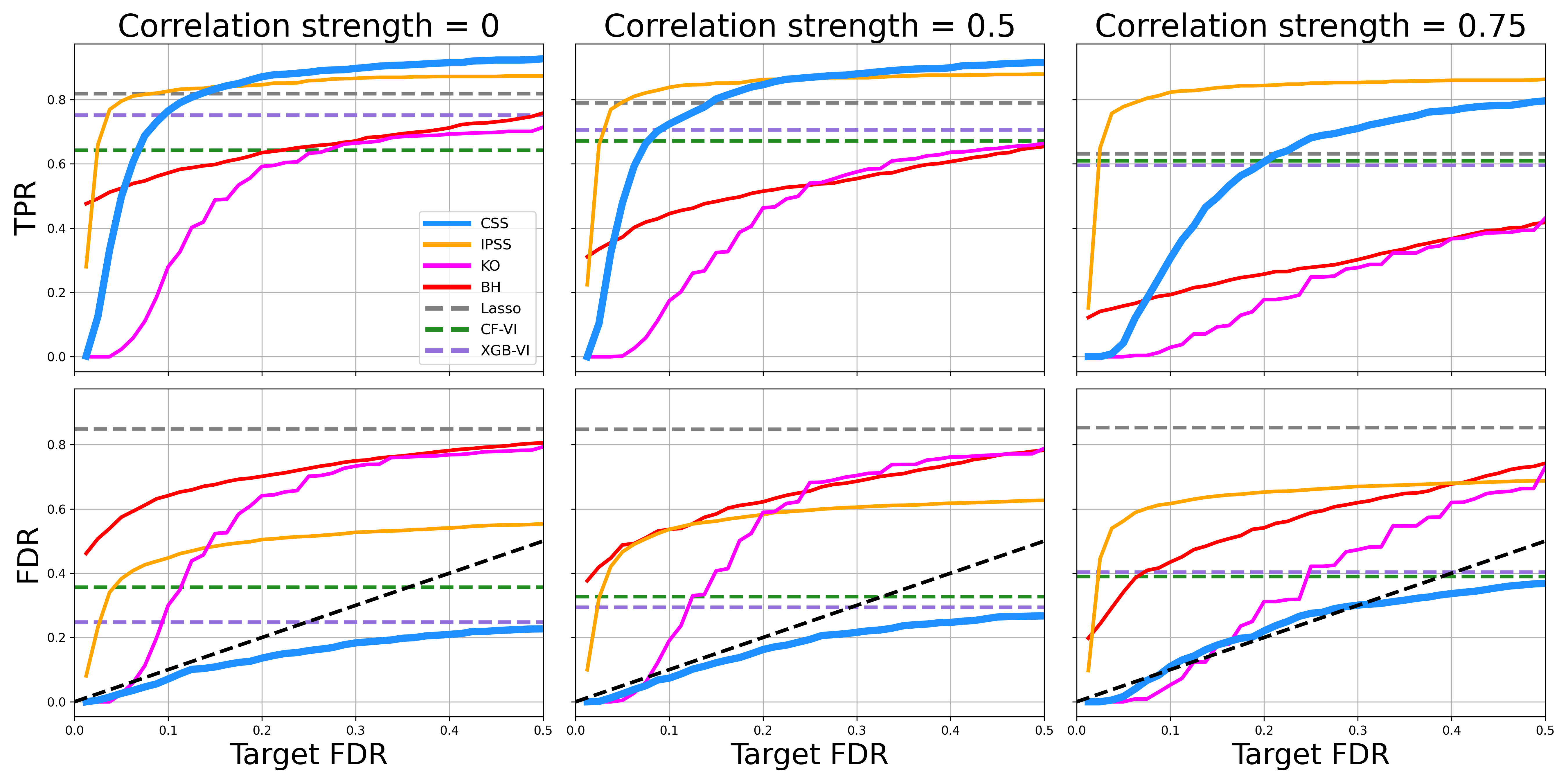}
\caption{\textit{Nonlinear results: covariate correlation, $\rho$.}}
\label{fig:sup_nonlinear_rho}
\end{figure}

\begin{figure}[H]
\includegraphics[width=\textwidth]{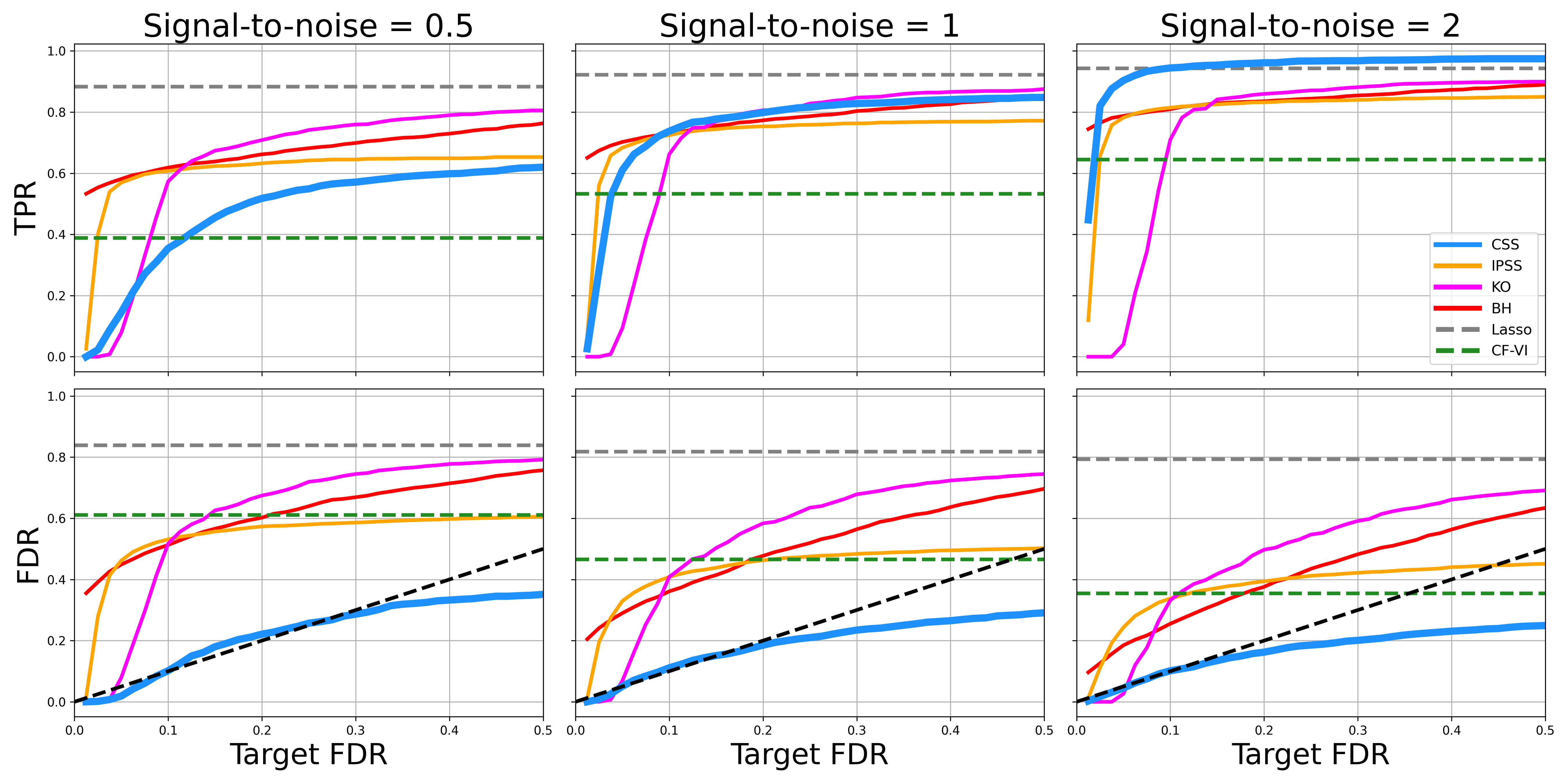}
\caption{\textit{Linear results: signal-to-noise ratio, $\mathrm{SNR}$.}}
\label{fig:sup_linear_snr}
\end{figure}

\begin{figure}[H]
\includegraphics[width=\textwidth]{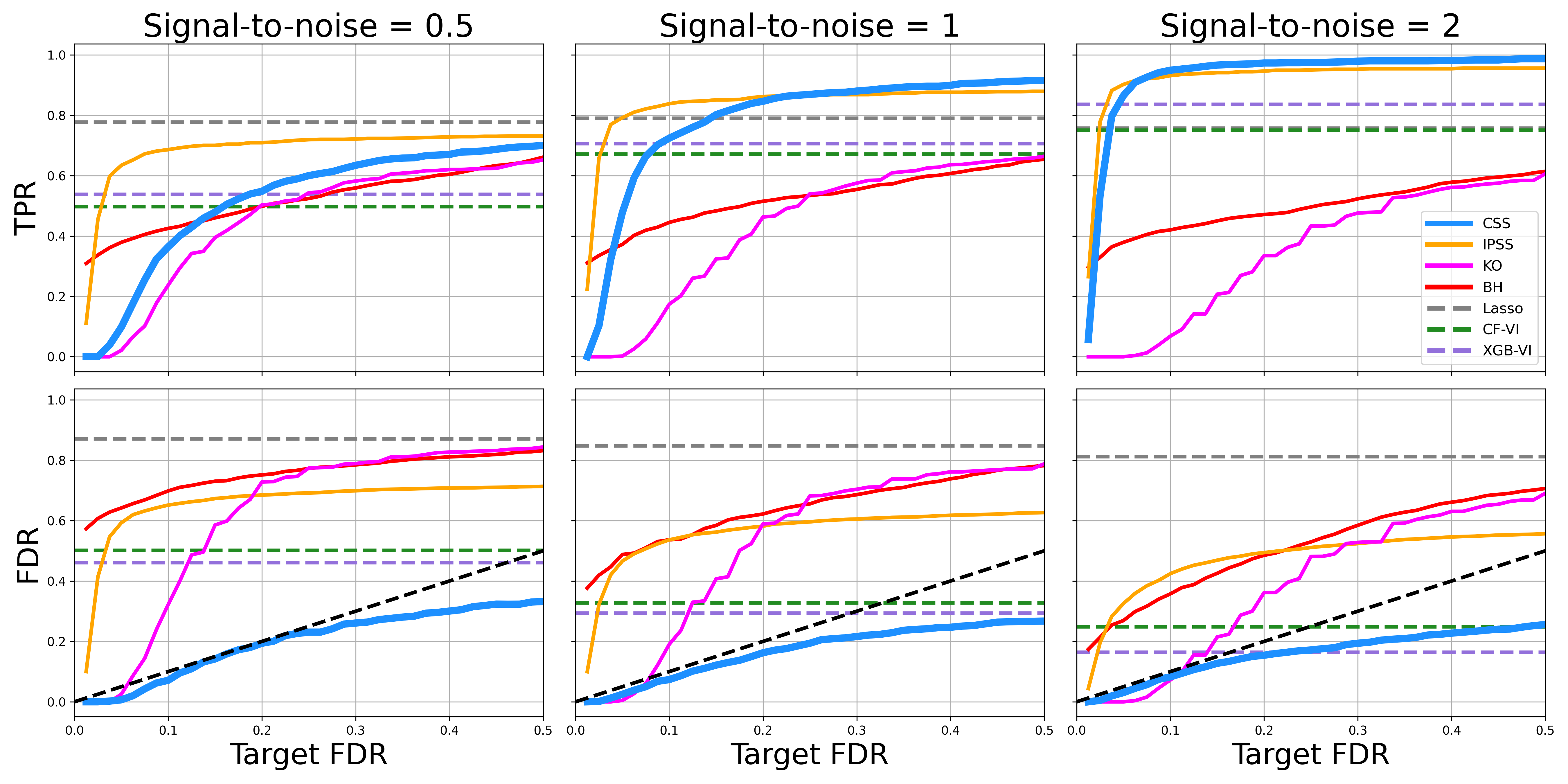}
\caption{\textit{Nonlinear results: signal-to-noise ratio, $\mathrm{SNR}$.}}
\label{fig:sup_nonlinear_snr}
\end{figure}

\begin{figure}[H]
\includegraphics[width=\textwidth]{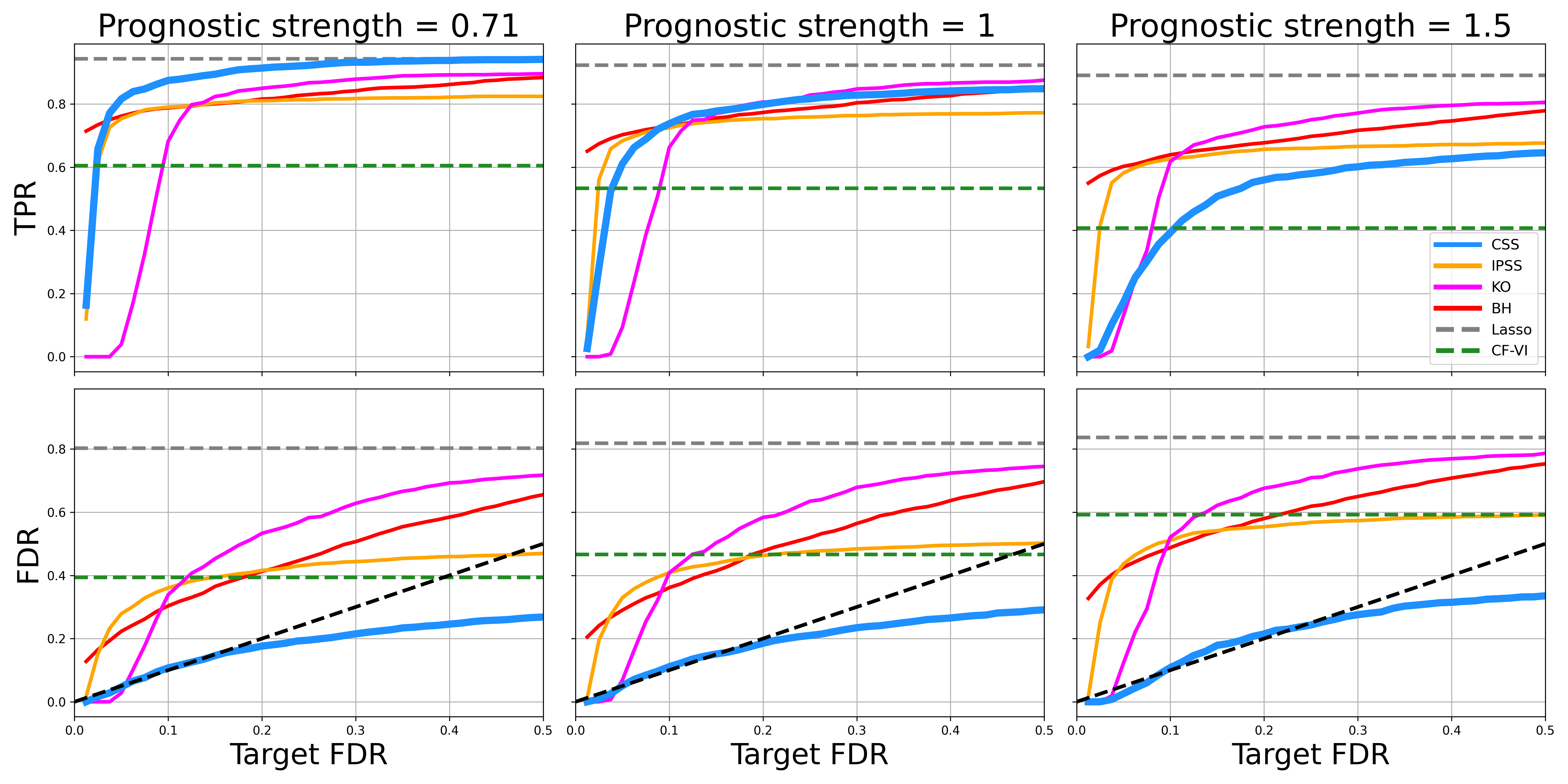}
\caption{\textit{Linear results: prognostic strength, $a$.}}
\label{fig:sup_linear_prog_strength}
\end{figure}

\begin{figure}[H]
\includegraphics[width=\textwidth]{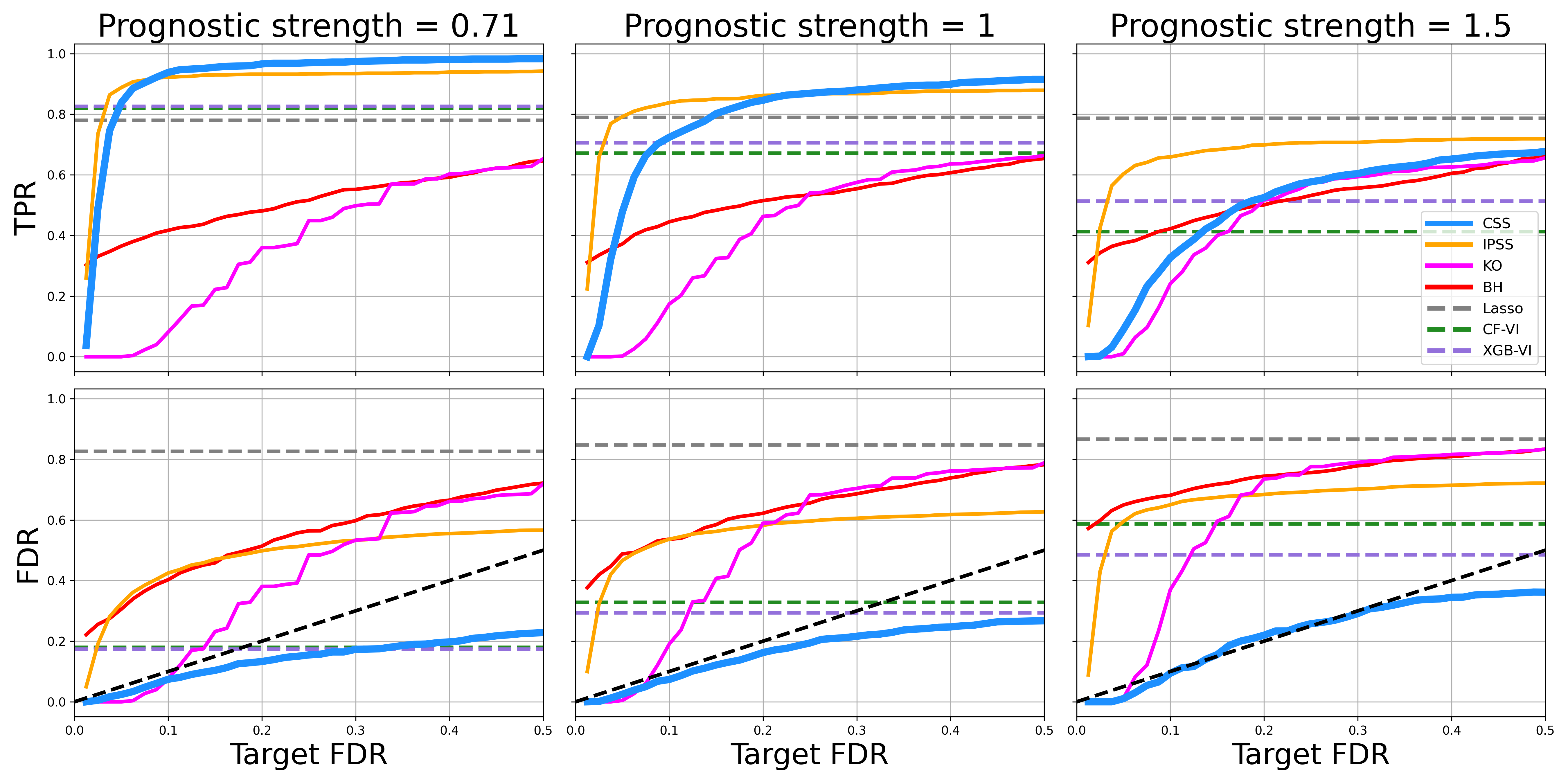}
\caption{\textit{Nonlinear results: prognostic strength, $a$.}}
\label{fig:sup_nonlinear_prog_strength}
\end{figure}

\clearpage

\section{Additional Application Details}\label{sup_sec:application}


\subsection{PRIME Randomized Trial}\label{sup_sec:prime}

\subsubsection{Covariate Correlation Structure}

The covariate matrix exhibits two clusters of mutation indicators with $r > 0.90$, namely BMMTR2/BMMTR15 and BMMTR5/BMMTR16. Such correlation structure resembles a setting in which standard selection procedures are known to over-flag.

\begin{figure}[H]
\centering
\includegraphics[width=0.95\textwidth]{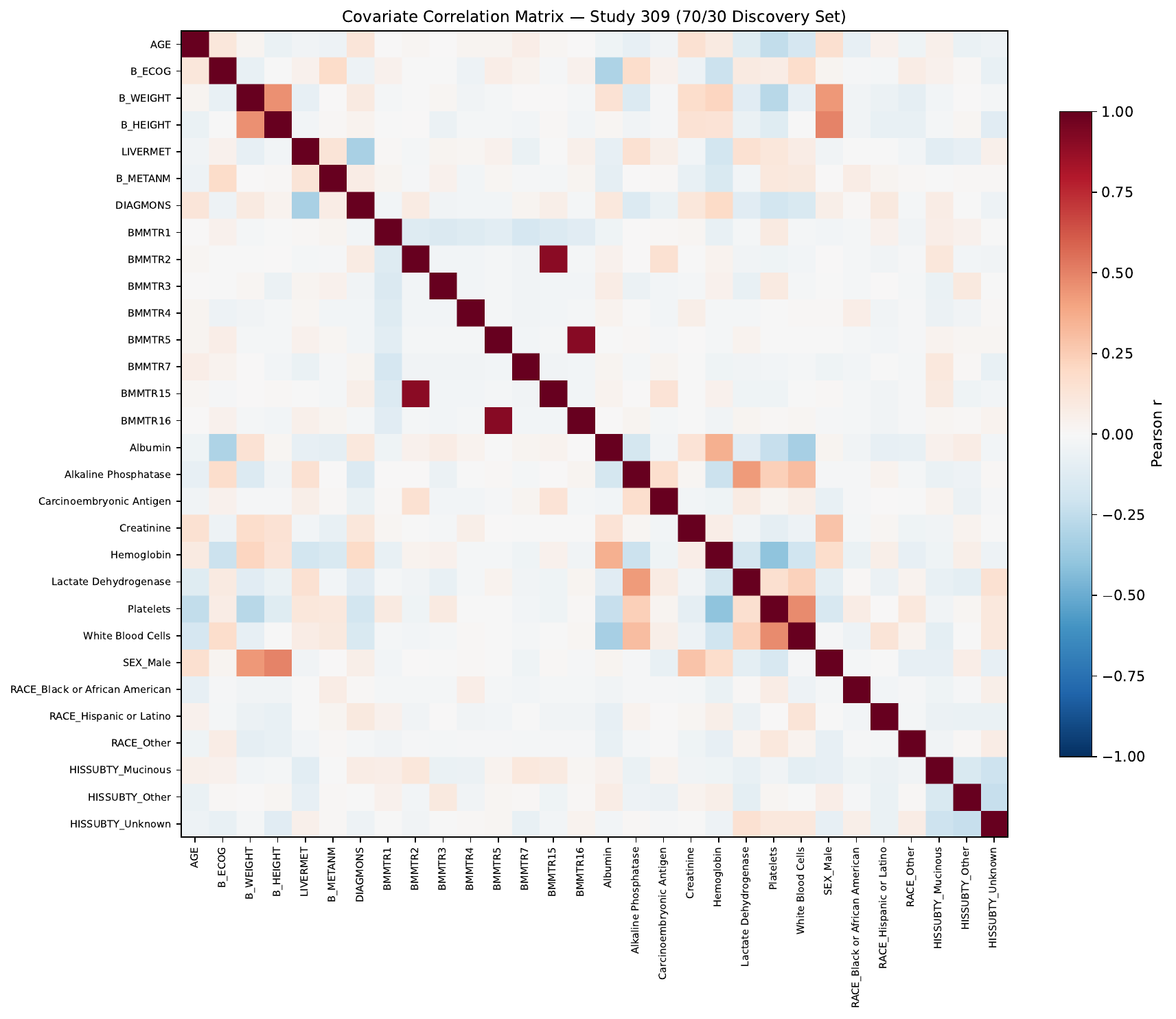}
\caption{Covariate correlation matrix on the discovery set ($n = 603$).}
\label{fig:corr_309}
\end{figure}

\begin{table}[H]
\centering
\begin{tabular}{llr}
\toprule
\textbf{Covariate 1} & \textbf{Covariate 2} & \textbf{$r$} \\
\midrule
BMMTR5 & BMMTR16 & $+0.911$ \\
BMMTR2 & BMMTR15 & $+0.905$ \\
B\_HEIGHT & SEX\_Male & $+0.500$ \\
Platelets & White Blood Cells & $+0.476$ \\
B\_WEIGHT & B\_HEIGHT & $+0.454$ \\
B\_WEIGHT & SEX\_Male & $+0.437$ \\
Alkaline Phosphatase & Lactate Dehydrogenase & $+0.426$ \\
Hemoglobin & Platelets & $-0.403$ \\
Albumin & Hemoglobin & $+0.359$ \\
Albumin & White Blood Cells & $-0.333$ \\
\bottomrule
\end{tabular}
\caption{Top pairwise correlations among covariates on the discovery set.}
\label{tab:309_correlations}
\end{table}

\subsubsection{Pre-Selection Analysis}

As a baseline reference for \texttt{CausalStabSel}, we run standard AIPW regressions on the discovery sample. AIPW pseudo-outcomes are computed via 5-fold cross-fitting with \texttt{LogisticRegressionCV} for both propensity and outcome models, propensity clipped to $[0.10, 0.90]$, and AIPW scores winsorized at the 1st and 99th percentiles.

\paragraph{Univariate AIPW regressions.}
Each covariate is tested individually: $\text{AIPW}_i = \alpha + \beta(X_j - \bar{X}_j) + \varepsilon_i$ with HC1 robust standard errors.

\begin{table}[H]
\centering
\small
\caption{Univariate pre-selection: covariates significant at $p < 0.05$.}
\label{tab:309_univariate}
\begin{tabular}{lcccc}
\toprule
& \multicolumn{2}{c}{\textbf{70/30} ($n=603$)} & \multicolumn{2}{c}{\textbf{80/20} ($n=689$)} \\
\cmidrule(lr){2-3} \cmidrule(lr){4-5}
\textbf{Covariate} & Coef & $p$ & Coef & $p$ \\
\midrule
Carcinoembryonic Antigen & $-0.0000$ & $0.0044$ & $-0.0000$ & $0.0185$ \\
BMMTR1 & $-0.1703$ & $0.0331$ & $-0.1508$ & $0.0447$ \\
Platelets & $-0.0004$ & $0.1926$ & $-0.0007$ & $0.0188$ \\
BMMTR3 & $-0.3537$ & $0.1686$ & $-0.3706$ & $0.0359$ \\
\midrule
\multicolumn{5}{l}{Total significant at 5\%: 2/30 (70/30), 4/30 (80/20)} \\
\bottomrule
\end{tabular}
\end{table}

\paragraph{Multivariate AIPW regressions.}
All 30 covariates included simultaneously: $\text{AIPW}_i = \alpha + \sum_{j=1}^{30} \beta_j (X_{ij} - \bar{X}_j) + \varepsilon_i$ with HC1 robust standard errors.

\begin{table}[H]
\centering
\small
\caption{Multivariate pre-selection: covariates significant at $p < 0.05$.}
\label{tab:309_multivariate}
\begin{tabular}{lcccc}
\toprule
& \multicolumn{2}{c}{\textbf{70/30} ($n=603$)} & \multicolumn{2}{c}{\textbf{80/20} ($n=689$)} \\
\cmidrule(lr){2-3} \cmidrule(lr){4-5}
\textbf{Covariate} & Coef & $p$ & Coef & $p$ \\
\midrule
BMMTR15 & $-0.9347$ & $<0.001$ & $-0.8875$ & $<0.001$ \\
BMMTR2 & $+0.9239$ & $0.0035$ & $+0.9062$ & $0.0030$ \\
Carcinoembryonic Antigen & $-0.0000$ & $0.0047$ & $-0.0000$ & $0.0271$ \\
BMMTR1 & $-0.2159$ & $0.0154$ & $-0.2118$ & $0.0069$ \\
Lactate Dehydrogenase & $-0.0001$ & $0.0194$ & $-0.0001$ & $0.0164$ \\
BMMTR4 & $-0.4676$ & $0.0461$ & $-0.4729$ & $0.0408$ \\
\midrule
\multicolumn{5}{l}{Total significant at 5\%: 6/30 (70/30), 7/30 (80/20)} \\
\bottomrule
\end{tabular}
\end{table}

Plain regression finds 2--4 covariates significant at 5\% univariately and 6--7 multivariately. Many of these associations are likely driven by the correlation structure documented above: in particular, BMMTR2 and BMMTR15 ($r = 0.91$) are both flagged in the multivariate analysis.

\subsubsection{Causal Stability Selection Results}

We report selection results separately for the two selectors and the two split ratios. Shaded rows denote features selected in both the 70/30 and 80/20 splits; unshaded rows denote features selected in only one of the two configurations.

\begin{table}[H]
\centering
\caption{Selected features --- DR-learner (Ridge) / LASSO ($q < 0.10$).}
\label{tab:309_qvals_main}
\begin{tabular}{lcccc}
\toprule
& \multicolumn{2}{c}{\textbf{$q$-values}} & \multicolumn{2}{c}{\textbf{Selected}} \\
\cmidrule(lr){2-3} \cmidrule(lr){4-5}
\textbf{Feature} & 70/30 & 80/20 & 70/30 & 80/20 \\
\midrule
 BMMTR1 & 0.039 & 0.032 & \yes & \yes \\
 BMMTR3 & 0.041 & 0.032 & \yes & \yes \\
 AGE & 0.039 & 0.091 & \yes & \yes \\
 Platelets & 0.039 & 0.032 & \yes & \yes \\
Lactate Dehydrogenase & 0.039 & 0.184 & \yes & \\
BMMTR4 & 0.102 & 0.093 & & \yes \\
\bottomrule
\end{tabular}
\par\vspace{0.3em}
\footnotesize 70/30: 5 selected. 80/20: 5 selected. The four-feature core (BMMTR1, BMMTR3, AGE, Platelets) is invariant; LDH and BMMTR4 swap depending on the split.
\end{table}

\begin{table}[H]
\centering
\caption{Selected features --- DR-learner (Ridge) / GB ($q < 0.10$).}
\label{tab:309_qvals_sens}
\begin{tabular}{lcccc}
\toprule
& \multicolumn{2}{c}{\textbf{$q$-values}} & \multicolumn{2}{c}{\textbf{Selected}} \\
\cmidrule(lr){2-3} \cmidrule(lr){4-5}
\textbf{Feature} & 70/30 & 80/20 & 70/30 & 80/20 \\
\midrule
 AGE & 0.017 & 0.041 & \yes & \yes \\
 BMMTR1 & 0.017 & 0.033 & \yes & \yes \\
 BMMTR3 & 0.017 & 0.033 & \yes & \yes \\
 Platelets & 0.017 & 0.033 & \yes & \yes \\
 B\_WEIGHT & 0.018 & 0.052 & \yes & \yes \\
 Lactate Dehydrogenase & 0.017 & 0.069 & \yes & \yes \\
Hemoglobin & 0.054 & 0.162 & \yes & \\
White Blood Cells & 0.175 & 0.072 & & \yes \\
\bottomrule
\end{tabular}
\par\vspace{0.3em}
\footnotesize 70/30: 7 selected. 80/20: 7 selected. The six-feature core is invariant across splits; Hemoglobin and WBC swap depending on the split. GB selects more features than LASSO at the same threshold.
\end{table}

Across both selectors and both split ratios, BMMTR1, BMMTR3, AGE, and Platelets are selected in all four configurations. BMMTR2 and BMMTR15, both flagged by multivariate pre-selection, are not selected by \texttt{CausalStabSel} under any configuration; the canonical KRAS-exon-2 and exon-4 indicators (BMMTR1, BMMTR3) are. This pattern (e.g., exclusion of demonstrably collinear proxies, retention of canonical effect modifiers) is consistent with the FDR-control behavior established in the simulations.

\subsubsection{Held-out Inference (70/30 split)}

For each selected feature $j$, we estimate the BLP of the CATE on $X_j$ on the inference sample using $\hat\tau(\bx_i) = \alpha + \beta_j(X_{j,i} - \bar{X}_j) + \varepsilon_i$ with HC1 robust standard errors.

\begin{table}[H]
\centering
\caption{Inference results, DR-learner / LASSO core, 70/30 split ($n_{\text{inf}} = 259$).}
\label{tab:309_inf_main_70}
\begin{tabular}{lrrr}
\toprule
\textbf{Feature} & \textbf{Coef} & \textbf{SE} & \textbf{95\% CI} \\
\midrule
 Platelets & $-0.0014$ & $0.0005$ & $[-0.002,\ -0.000]$ \\
BMMTR1 & $-0.1198$ & $0.1232$ & $[-0.363,\ 0.123]$ \\
BMMTR3 & $+0.0729$ & $0.2426$ & $[-0.405,\ 0.551]$ \\
AGE & $+0.0075$ & $0.0056$ & $[-0.004,\ 0.019]$ \\
Lactate Dehydrogenase & $-0.0001$ & $0.0001$ & $[-0.0003,\ 0.0001]$ \\
\midrule
\multicolumn{4}{l}{ATE $= -0.016$ (SE $= 0.065$); 95\% CI $[-0.144,\ 0.112]$} \\
\bottomrule
\end{tabular}
\end{table}

The estimated ATE is essentially zero, consistent with the original PRIME finding that benefit was confined to the KRAS wild-type subgroup. Among the selected modifiers, only platelets has a 95\% CI excluding zero. Treatment-effect heterogeneity coexists with a null average effect, the regime in which discovery of effect modifiers is most informative.

\subsection{Birthweight Observational Study}\label{sup_sec:birthweight}

\subsubsection{Covariate Correlation Structure}

\begin{figure}[H]
\centering
\includegraphics[width=0.85\textwidth]{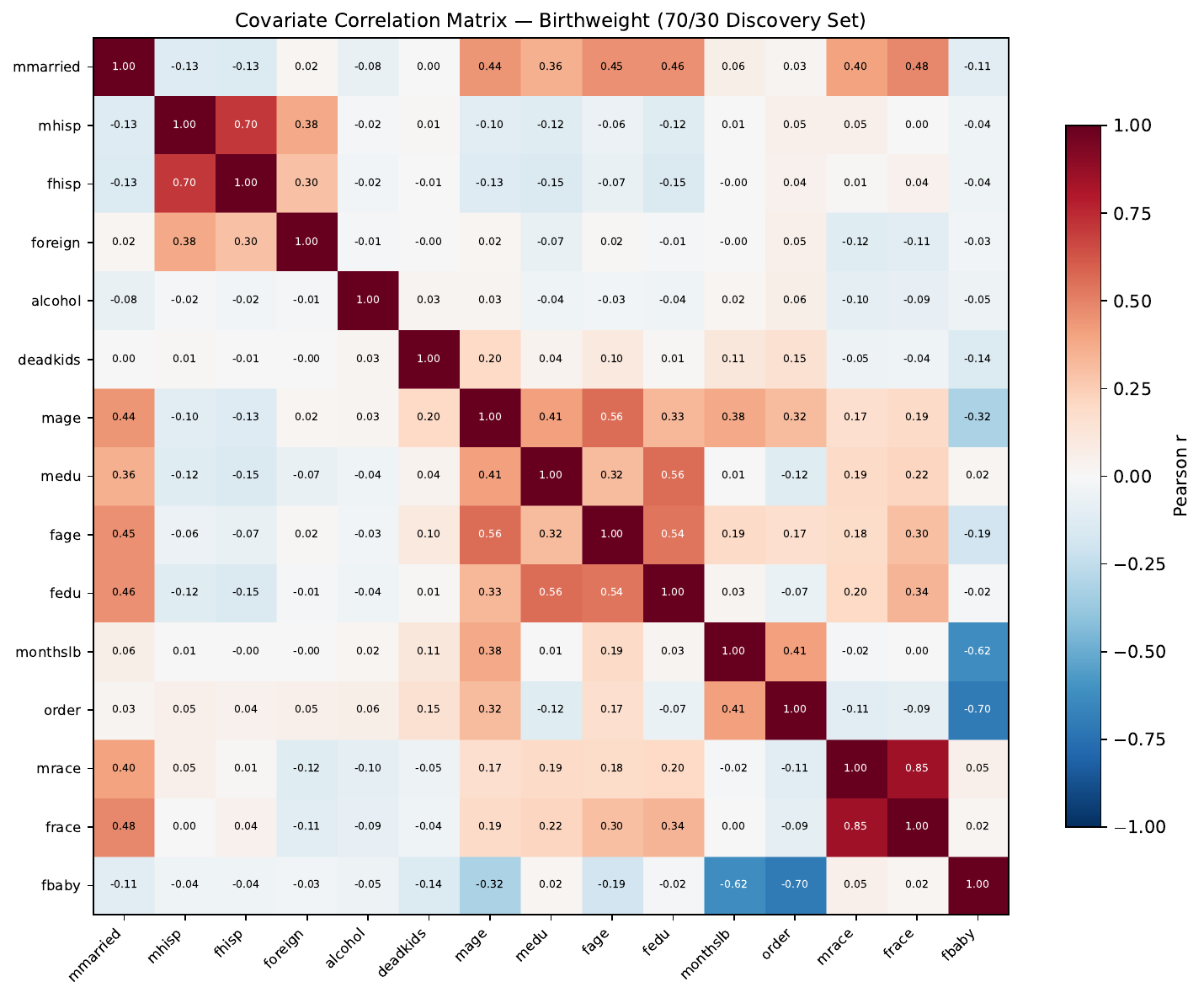}
\caption{Covariate correlation matrix on the 70/30 discovery set ($n = 3{,}249$).}
\label{fig:corr_bw}
\end{figure}

\begin{table}[H]
\centering
\caption{Top pairwise correlations among covariates on the 70/30 discovery set.}
\label{tab:bw_correlations}
\begin{tabular}{llr}
\toprule
\textbf{Covariate 1} & \textbf{Covariate 2} & \textbf{$r$} \\
\midrule
Mother White & Father White & $+0.848$ \\
Mother Hispanic & Father Hispanic & $+0.704$ \\
Birth order & First baby & $-0.704$ \\
Months since last birth & First baby & $-0.623$ \\
Mother education & Father education & $+0.557$ \\
Mother age & Father age & $+0.556$ \\
Married & Father White & $+0.479$ \\
Married & Father education & $+0.459$ \\
Married & Father age & $+0.455$ \\
Married & Mother age & $+0.439$ \\
\bottomrule
\end{tabular}
\end{table}

The covariate matrix exhibits two distinct correlation clusters: assortative mating between maternal and paternal characteristics (Mother/Father White, $r = 0.85$; Mother/Father Hispanic, $r = 0.70$; parental education and age, $r \approx 0.56$), and mechanically related parity variables (Birth order / First baby, $r = -0.70$; Months since last birth / First baby, $r = -0.62$). Marital status correlates with multiple paternal characteristics ($r \approx 0.44$--$0.48$). These correlations are substantially higher than in the PRIME data and may drive spurious effect modifier discovery in standard regression.

\subsubsection{Pre-Selection Analysis}

As a baseline reference for \texttt{CausalStabSel}, we run standard AIPW regressions on the discovery sample. AIPW pseudo-outcomes are computed via 5-fold cross-fitting with \texttt{LogisticRegressionCV} for the propensity model and \texttt{RidgeCV} for the outcome model, propensity clipped to $[0.10, 0.90]$, and AIPW scores winsorized at the 1st and 99th percentiles.

\paragraph{Univariate AIPW regressions.}

Each covariate is tested individually: $\text{AIPW}_i = \alpha + \beta(X_j - \bar{X}_j) + \varepsilon_i$ with HC1 robust standard errors.

\begin{table}[H]
\centering
\small
\caption{Univariate pre-selection: covariates significant at $p < 0.10$.}
\label{tab:bw_univariate}
\begin{tabular}{lcccc}
\toprule
& \multicolumn{2}{c}{\textbf{70/30} ($n=3{,}249$)} & \multicolumn{2}{c}{\textbf{80/20} ($n=3{,}713$)} \\
\cmidrule(lr){2-3} \cmidrule(lr){4-5}
\textbf{Covariate} & Coef (g) & $p$ & Coef (g) & $p$ \\
\midrule
Mother age & $-9.8$ & $0.010$ & $-13.1$ & $<0.001$ \\
First baby & $+106.6$ & $0.011$ & $+128.6$ & $0.001$ \\
Mother White & $-136.8$ & $0.030$ & $-147.1$ & $0.012$ \\
Birth order & $-36.5$ & $0.052$ & $-52.2$ & $0.005$ \\
Father White & $-107.6$ & $0.065$ & $-111.7$ & $0.043$ \\
Married & $-92.9$ & $0.063$ & $-101.3$ & $0.032$ \\
Months since last birth & $-1.3$ & $0.075$ & $-1.4$ & $0.034$ \\
\midrule
\multicolumn{5}{l}{Significant at 5\%: 3/15 (70/30), 7/15 (80/20). At 10\%: 7/15 (70/30), 8/15 (80/20).} \\
\bottomrule
\end{tabular}
\end{table}

\paragraph{Multivariate AIPW regressions.}
All 15 covariates included simultaneously: $\text{AIPW}_i = \alpha + \sum_{j=1}^{15} \beta_j (X_{ij} - \bar{X}_j) + \varepsilon_i$ with HC1 robust standard errors.

\begin{table}[H]
\centering
\small
\caption{Multivariate pre-selection summary.}
\label{tab:bw_multivariate}
\begin{tabular}{lcc}
\toprule
& \textbf{70/30} ($n=3{,}249$) & \textbf{80/20} ($n=3{,}713$) \\
\midrule
Significant at 5\% & 0/15 & 0/15 \\
Significant at 10\% & 0/15 & 1/15 (Mother White, $p=0.082$) \\
\bottomrule
\end{tabular}
\end{table}

Plain regression on the discovery set finds 3--7 covariates significant at 5\% univariately, but \emph{no} covariates survive multivariate adjustment in either split. This pattern -- many univariate hits collapsing in multivariate analysis -- is a well-known artifact of correlated covariate structure: shared variance is partialled out when covariates enter the regression jointly, deflating individual $t$-statistics. The contrast with the PRIME data, where multivariate regression revealed \emph{more} signals than univariate, illustrates that standard regression can fail in opposite directions depending on the correlation geometry of the design.

\subsubsection{Causal Stability Selection Results}

We report selection results for the two selectors and the two split ratios. Shaded rows denote features selected in both the 70/30 and 80/20 splits; unshaded rows denote features selected in only one of the two configurations.

\begin{table}[H]
\centering
\caption{Selected features --- DR-learner (Ridge) / LASSO ($q < 0.10$).}
\label{tab:bw_qvals_main}
\begin{tabular}{lcccc}
\toprule
& \multicolumn{2}{c}{\textbf{$q$-values}} & \multicolumn{2}{c}{\textbf{Selected}} \\
\cmidrule(lr){2-3} \cmidrule(lr){4-5}
\textbf{Feature} & 70/30 & 80/20 & 70/30 & 80/20 \\
\midrule
 Mother age & 0.053 & 0.040 & \yes & \yes \\
 First baby & 0.053 & 0.040 & \yes & \yes \\
\bottomrule
\end{tabular}
\par\vspace{0.3em}
\footnotesize Both splits select the same two features.
\end{table}

\begin{table}[H]
\centering
\caption{Selected features --- DR-learner (Ridge) / GB ($q < 0.10$).}
\label{tab:bw_qvals_sens}
\begin{tabular}{lcccc}
\toprule
& \multicolumn{2}{c}{\textbf{$q$-values}} & \multicolumn{2}{c}{\textbf{Selected}} \\
\cmidrule(lr){2-3} \cmidrule(lr){4-5}
\textbf{Feature} & 70/30 & 80/20 & 70/30 & 80/20 \\
\midrule
 Mother age & 0.028 & 0.027 & \yes & \yes \\
 Birth order & 0.028 & 0.027 & \yes & \yes \\
 Foreign-born & 0.028 & 0.036 & \yes & \yes \\
 Mother White & 0.042 & 0.043 & \yes & \yes \\
 Months since last birth & 0.028 & 0.041 & \yes & \yes \\
First baby & \na & 0.033 & & \yes \\
Married & \na & 0.039 & & \yes \\
Father age & 0.028 & \na & \yes & \\
Father education & 0.063 & \na & \yes & \\
Mother education & 0.089 & \na & \yes & \\
\bottomrule
\end{tabular}
\par\vspace{0.3em}
\footnotesize 70/30: 8 selected. 80/20: 7 selected. The five-feature core is invariant across splits; Father age, Father education, and Mother education appear only in 70/30, while First baby and Married appear only in 80/20.
\end{table}

The two selectors differ in selectivity: LASSO retains two features (Mother age, First baby), while GB retains seven to eight, including the parity correlate Birth order ($r = -0.70$ with First baby) and the demographic correlates Mother White, Foreign-born, and Months since last birth. Mother age is selected in all four configurations of split ratio and selector. The parity signal appears as First baby under LASSO and as Birth order under GB; the two indicators are mechanically related and capture the same underlying phenomenon. Both selectors correctly avoid Father White ($r = 0.85$ with Mother White), an obvious assortative-mating proxy that would inflate naive multivariate regression.

\subsubsection{Held-out Inference}

For each selected feature $j$, we estimate the BLP of the CATE on $X_j$ on the inference sample using $\hat\tau(\bx_i) = \alpha + \beta_j(X_{j,i} - \bar{X}_j) + \varepsilon_i$ with HC1 robust standard errors.

\begin{table}[H]
\centering
\caption{Inference results, DR-learner / LASSO selections, 70/30 split ($n_{\text{inf}} = 1{,}393$).}
\label{tab:bw_inf_main_70}
\begin{tabular}{lrrr}
\toprule
\textbf{Feature} & \textbf{Coef (g)} & \textbf{SE} & \textbf{95\% CI} \\
\midrule
 Mother age & $-14.9$ & $5.8$ & $[-26.3,\ -3.5]$ \\
 First baby & $+145.8$ & $65.7$ & $[17.0,\ 274.6]$ \\
\midrule
\multicolumn{4}{l}{ATE $= -266.4$g (SE $= 32.7$); 95\% CI $[-330.5,\ -202.3]$} \\
\bottomrule
\end{tabular}
\end{table}

\end{document}